\documentclass[10pt,twocolumn]{IEEEtran} 	

\usepackage{fancyhdr, epsfig, epsf, amsthm, amsmath, amssymb, amsfonts, subfigure, color}
\usepackage{threeparttable}
\usepackage[noadjust]{cite}
\usepackage{dsfont}
\usepackage{enumerate}
\usepackage{comment}
\usepackage{soul}
\usepackage{algorithm}
\usepackage{algpseudocode}

\newtheorem{definition}{Definition}

\newtheorem{lemma}{Lemma}
\newtheorem{corollary}{Corollary}

\newtheorem{proposition}{Proposition}
\newtheorem{remark}{Remark}

\newtheorem{assumption}{Assumption}

\def\l{\left}
\def\r{\right}
\def\({\left(}
\def\){\right)}

\def\b0{{\mathbf{0}}}

\includecomment{comment}	

\def\papertitle{\Huge Live Prefetching for \\ Mobile Computation Offloading}

\IEEEoverridecommandlockouts 

\begin{document}

\title{ \fontsize{21}{21}\selectfont \papertitle}
\author{Seung-Woo Ko,   Kaibin Huang,  Seong-Lyun Kim, and Hyukjin Chae
\thanks{S.-W.~Ko and K.~Huang are with the Department of Electrical and Electronic Engineering, The University of Hong Kong, Pok Fu Lam, Hong Kong (e-mail: swko@eee.hku.hk, haungkb@eee.hku.hk).
S.-L.~Kim is with the School of Electrical and Electronic Engineering,
Yonsei University, Seoul, Korea
(email: slkim@ramo.yonsei.ac.kr).  H.~Chae is with LG Electronics, Seoul, Korea (e-mail: hyukjin.chae@lge.com). 
}
\thanks{This work was supported by LG Electronics and the National Research Foundation of Korea under Grant NRF-2014R1A2A1A11053234. 
Part of this work was presented at IEEE ICC 2017.}}
\maketitle

\begin{abstract}
The conventional designs of mobile computation offloading fetch user-specific data to the cloud prior to computing, called \emph{offline prefetching}.
However, this approach can potentially result in excessive fetching of large volumes of data and cause heavy loads on radio-access networks. To solve this problem,  the novel technique of \emph{live prefetching} is proposed in this paper that seamlessly integrates the task-level computation prediction and prefetching within the cloud-computing process of a large program with numerous tasks. The technique avoids excessive fetching but retains the feature of leveraging prediction to reduce the program runtime and mobile transmission energy. By modeling the tasks in an offloaded program as a stochastic sequence, stochastic optimization  is applied to design fetching policies to minimize mobile energy consumption under a deadline constraint. The policies  enable real-time  control  of the prefetched-data sizes of candidates for future tasks. For slow fading, the optimal policy is derived and shown to have a threshold-based structure, selecting candidate tasks for prefetching and controlling their prefetched data based  on their likelihoods. The result is extended to design close-to-optimal prefetching policies to fast fading channels. Compared with fetching without prediction, live prefetching is shown theoretically to always achieve reduction on mobile energy consumption. 
\end{abstract}


\section{Introduction}

In the past decade, mobile devices have become the primary platforms for computing and Internet access. Despite their limited computation resources and power supply,  mobile devices are expected to support complex  applications such as multimedia processing, online gaming, virtual reality and sensing \cite{GSMA5G}. A promising technology for resolving this conflict is \emph{mobile computation offloading} (MCO) that offloads computation intensive tasks from mobiles to the cloud, thereby lengthening  their battery lives  and enhancing their computation capabilities \cite{Kumar2010}. In the future, MCO may enable wearable computing devices to perform sophisticated functions currently~feasible only on larger devices such as smartphones and tablet computers.  Nevertheless, MCO requires fetching user-specific data from mobiles to the cloud and the consumed transmission energy  offsets the energy gain due to offloading. To address this issue, we propose the use of computation prediction to reduce transmission energy consumption. Specifically, \emph{live prefetching} techniques are designed for predicting subsequent tasks in a offloaded program and fetching the needed data in advance. The technique reduces not only the latency but also the energy consumption by  increasing the fetching duration and exploiting opportunistic transmission. 
 
 \subsection{Prior Work}

Offloading a task from a mobile to the cloud reduces the load of the local CPU and hence the  energy consumption  for  mobile computing. However, MCO increases the (mobile) transmission-energy consumption \cite{Kumar2010}. Optimizing the tradeoff with the criterion of minimum mobile-energy consumption is the main theme for designing the MCO algorithms \cite{Zhang2013TVT, Kwak2015, Zhang2015, Kao2014, Yang2013}.  In \cite{Zhang2013TVT}, algorithms are designed for energy-efficient transmission and CPU control; then comparing the resultant energy consumption for MCO and mobile computing gives the optimal offloading decision. A different approach is proposed  in \cite{Kwak2015} that applies Lyapunov optimization theory to design energy-efficient MCO. Offloading a complex program consisting of multiple tasks gives rise to the research issue of optimal program partitioning (for offloading and local computing). Schemes for dynamic program   partitioning  have been developed based on different program models such as the  linear-task model \cite{Zhang2015}, the  tree-task model~\cite{Kao2014} and  the model of data streaming \cite{Yang2013}.

In recent research, more complex MCO systems have been designed \cite{You2016, Sardellitti2015, Chen2016, Kaewpuang2013, Azimi2016, Xiang2014, Zhao2015, Wang2014, You2016JSAC, Mao2016 }. The joint computation-and-radio resource allocation is studied in \cite{You2016} and \cite{Sardellitti2015} for a multiuser MCO system and a multi-cell MCO system, respectively. Distributed MCO systems are designed in \cite{Chen2016} using game theory.  The cooperative resource sharing framework for multiple clouds   
is proposed in \cite{Kaewpuang2013} for the revenue maximization of mobile cloud service providers 
under  the quality-of-service requirements of mobile applications.
A MCO system supporting multiple service levels is proposed in \cite{Azimi2016} where the superposition coding is proposed to facilitate the adaptation of the service level to a wireless channel. 
In \cite{Xiang2014}, an energy efficient offloading strategy is designed 
via the joint optimization of the  adaptive LTE/WiFi link selection and  data-transmission scheduling. 
The optimal scheduling policy fot MCO in the heterogenous cloud networks is studied in \cite{Zhao2015} where central and edge clouds coexist.
A service migration across edge clouds is investigated in \cite{Wang2014} that applies Markov decision process to design the optimal policy under the random-walk~mobility model.   
Last, energy harvesting and MCO serve the similar purpose of increasing mobile battery lives. This motivates the integration of these two technologies via algorithmic design in \cite{You2016JSAC} and \cite{Mao2016} where MCO is powered by wireless power transfer and ambient-energy harvesting, respectively. 

The MCO process considered in the aforementioned prior work requires fetching user-specific data from mobiles to the cloud before  cloud computing, called   \emph{offline  prefetching}. 
Its principle is identical to that of caching, namely delivering data to where is needed before its consumption, but targets mostly  the applications in cloud computing. Offline prefetching techniques have been designed to reduce latency in different applications and systems including MCO downlink \cite{Shu2013}, opportunistic fetching of multimedia data \cite{Master2016},  vehicular cloud systems \cite{Kim2016}, and wireless video streaming \cite{Hong2015}. Offline prefetching  is  impractical for running complex  mobile applications in the cloud such as real-time gaming or virtual reality. The user-specific data depends on the users' spontaneous behaviors and dynamic environment, and is thus difficult to predict. 
Even if data prediction is possible, offline prefetching for  complex applications can potentially result in transmission of enormous volumes of data over the air, placing a heavy burden on wireless networks. 
This motivates the development of the \emph{live prefetching} technique in the current work that performs real-time computation prediction  to avoid unnecessary prefetching for those tasks in an offloaded program that are unlikely to be executed.  Thereby, the technique reduces  the data rates for wireless transmission and the mobile transmission-energy consumption.  

\subsection{Contributions and Organization}
In this paper, we consider the MCO system where an access point connected to the cloud fetches user-specific data from a mobile to run an offloaded program  comprising a large set of potential tasks, called the \emph{task space}. Running the  program requires the execution of a fixed number  of  sequential tasks.  The tasks are  assumed to be unknown to the cloud till their execution and modeled as a sample path over  a Markov chain in the task space. Based on this model, during the execution of Task $(K-1)$, it is known to the cloud that there exists a set of candidates for Task $K$ with given likelihoods (specified by task transition probabilities). Furthermore, the cloud has the knowledge of the sizes of user-specific data required for individual tasks.  The task execution time is assumed to be uniform  for simplicity. 

During the execution of a particular task, a prefetcher in the cloud dynamically identifies a subset of candidates for the subsequent task and fetches parts of their input data. 
The remaining input data of the subsequent task is fetched after it is determined, called demand-fetching. Based on the above model, we propose the architecture for live prefetching, referred to simply as fetching hereafter, as follows. The program runtime is divided into fixed durations with each further divided into  
two phases for prefetching and demand-fetching. Specifically, considering the $K$-th duration, the prefetching phase is used for simultaneous prefetching for Task $K$ and  execution of Task $(K-1)$,  and the next phase for demand-fetching for Task~$K$. In this duration, the prefetching policy determines the sizes of prefetched data of candidate tasks for Task~$K$, by which the size of demand-fetched data for Task~$K$ is also  determined. 

This work focuses on optimizing the prefetching policies that dynamically select tasks for prefetching and control the corresponding prefetched data sizes under the criterion of minimum mobile energy consumption. The key results are summarized as follows. 

\begin{itemize}
\item (Prefetching Over Slow Fading Channels) Consider slow fading where the channel gain stays constant throughout the runtime of the offloaded program. The optimal prefetching policy  is derived in closed form. The policy has a threshold-based structure where the user-specific data for a candidate task is prefetched  only if its likelihood exceeds a given threshold and the prefetched data size is a monotone increasing function of the likelihood. 

\item (Prefetching Over Fast  Fading Channels) Consider fast  fading where the channel gain is independent and identically  distributed  (i.i.d.) over slots dividing each fetching duration. Given causal \emph{channel state information} (CSI) at the cloud, the derivation of the optimal prefetching policy is intractable due to the difficulty in determining the set of candidate tasks for prefetching. Conditioned on a given set, the optimal policy is shown to have a similar threshold-based structure as the slow-fading counterpart. This result is leveraged to design two sub-optimal policies based on estimating  the candidate-task set. 
 The resultant prefetching policies are shown by simulation to be close-to-optimal. 

\item (Prefetching Gain) Let the \emph{prefetching gain} be defined as the transmission-energy ratio between the cases without  and with prefetching. For slow fading, the gain is shown to be \emph{strictly} larger than one. The prefetching gain for fast fading is derived that has the same property as mentioned for the gain for slow fading. 
\end{itemize}

The rest of this paper is organized as follows.  
The system model is presented in Section~\ref{Section:SystemModel}. 
The live prefetching architecture and problem formulation are presented in Section~\ref{Section:ArchitectureandFormulation}.  
The prefetching policies are designed and analyzed for the cases of slow and fading fading 
in Sections~\ref{Section:SlowFadingPart} and \ref{Section:FastFadingPart}, respectively.  
Simulation results are presented in  Section \ref{Section:Simulation} following by concluding remarks  in Section \ref{Section:Conclusion}. 
 
\section{System Model}\label{Section:SystemModel}

\begin{figure}[t]
\centering
\includegraphics[width=8.5cm]{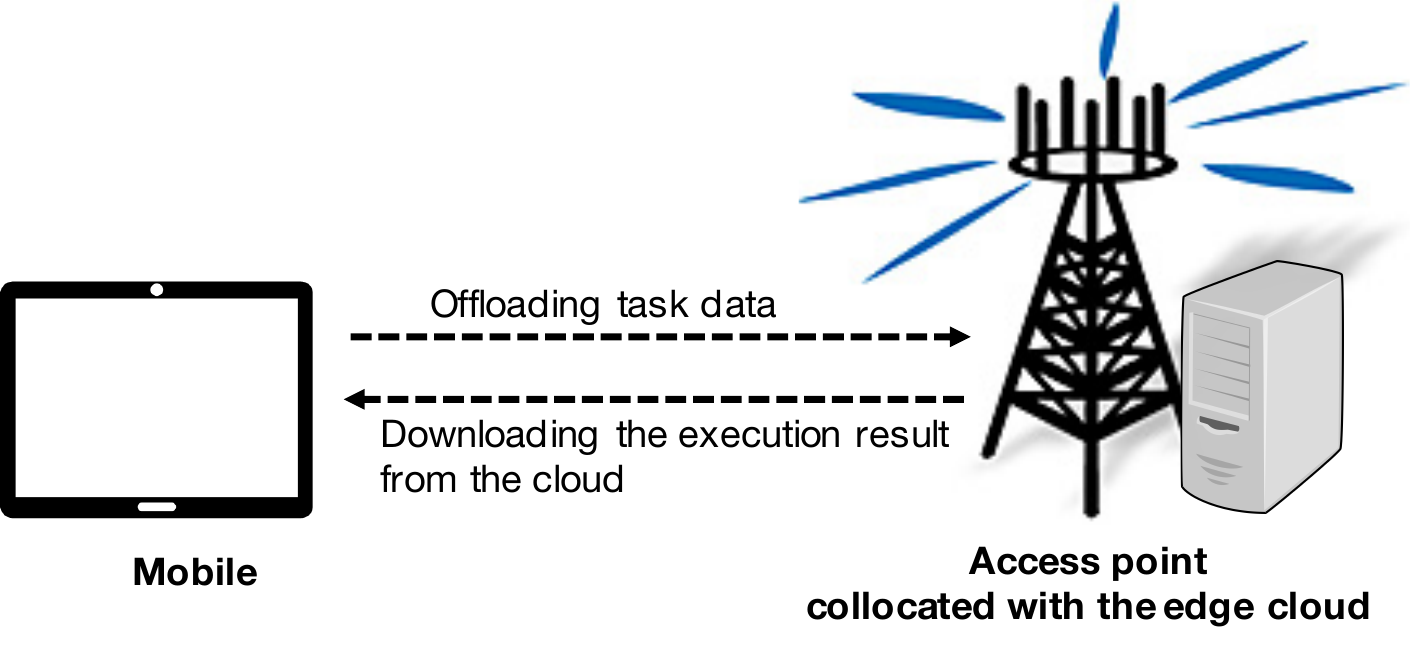}
\caption{MCO system with prefetching.}
\label{Fig:System}
\end{figure}

Consider the MCO system in Fig.~\ref{Fig:System}, which comprises a mobile and an access point connected to an edge cloud. 
The mobile attempts to execute a program comprising multiple tasks, each of which refers to the minimum computation unit that cannot be split.
For simplicity, the tasks are assumed to follow a sequential order.  
We consider the scenario where the tasks are computation intensive or have stringent deadlines such that they cannot be executed at the mobile due to its limited computation capabilities.
To overcome this difficulty, the  mobile  offloads the tasks to the edge cloud and thereby the application is remotely executed\footnote{It is assumed that the required input data for mobile computing is either cached at the AP or fetched from the user, which can be justified for typical mobile applications such as image processing or online gaming. The assumption simplifies the design by making the performance of backhaul links irrelevant.}.
For this purpose, the mobile transmits data  to and receive computation output via the access point.

The channel is modeled as follows.  The channel gain in slot  $n$ is denoted as $g_{n}$ with $g_n > 0$. We consider both slow and fast fading for the wireless channel. For slow fading, the channel gain is constant denoted as $g$:  $g_n = g$ for all $n$. Next, fast fading is modeled as  block fading where channel gains are constant over one time slot and i.i.d. 
over different slots. The fast fading model addresses the fact that a mission of next-generation wireless system is to support mobile computing for users with high mobility e.g., users in cars or on trains. The variation of channel gain during the computing process causes the required power for transporting a fixed number of bits in a single time slot to also vary with time. 
This makes the optimal strategy for joint adaptive transmission and prefetching more complex than the slow-fading counterpart.

\begin{assumption}[Causal CSI] \emph{Both the mobile and cloud have  perfect knowledge of the constant channel gain in the case of  slow fading and  of the channel gains in the current and past  slots as well as their distribution for the case of fast fading. }
\end{assumption}

Following the models in \cite{Zafer2009, Zhang2013TVT}, the energy consumption for transmitting  $b_{n}$ bits to the cloud in slot  $n$ is modeled using a convex monomial function, denoted as $\mathcal{E}$,  as follows: 
\begin{eqnarray}\label{EnergyModel}
\mathcal{E}(b_{n},g_{n})=\lambda \frac{(b_{n})^{m}}{g_{n}}
\end{eqnarray}
where the constants   $m$ and $\lambda$ represent the monomial order and the energy coefficient, respectively\footnote{It is implicitly assumed that an adaptive modulation and coding scheme is used with a continuous rate adapted to the channel gain. It is interesting to extend the current designs to practical cases e.g., fixed-rate channel codes combined with automatic repeat request (ARQ), which is outside the scope of current work.}.  
The monomial order is a positive integer depending on the specific  modulation-and-coding scheme and takes on values in the typical range 
of $2 \leq m \leq 5$.  Without loss of generality, we set  $\lambda=1$ to simplify notation, which has no effect on the  optimal prefetching policy.

\section{Live Prefetching: Architecture and Problem Formulation}\label{Section:ArchitectureandFormulation} 

In this section, the novel architecture of live prefetching is proposed 
where prefetching the next task and computing the current task 
are performed simultaneously.       
Based on the framework, the problem of optimal prefetching policy is formulated. 

\subsection{Live Prefetching Architecture}

\begin{figure}[t]
\centering
\subfigure[Task topology for the offloaded program]{\includegraphics[width=2.5in]{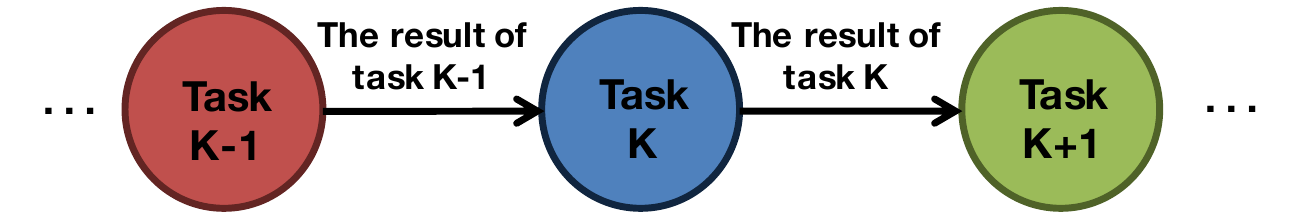}}\\
\subfigure[Live prefetching architecture]{\includegraphics[width=3.5in]{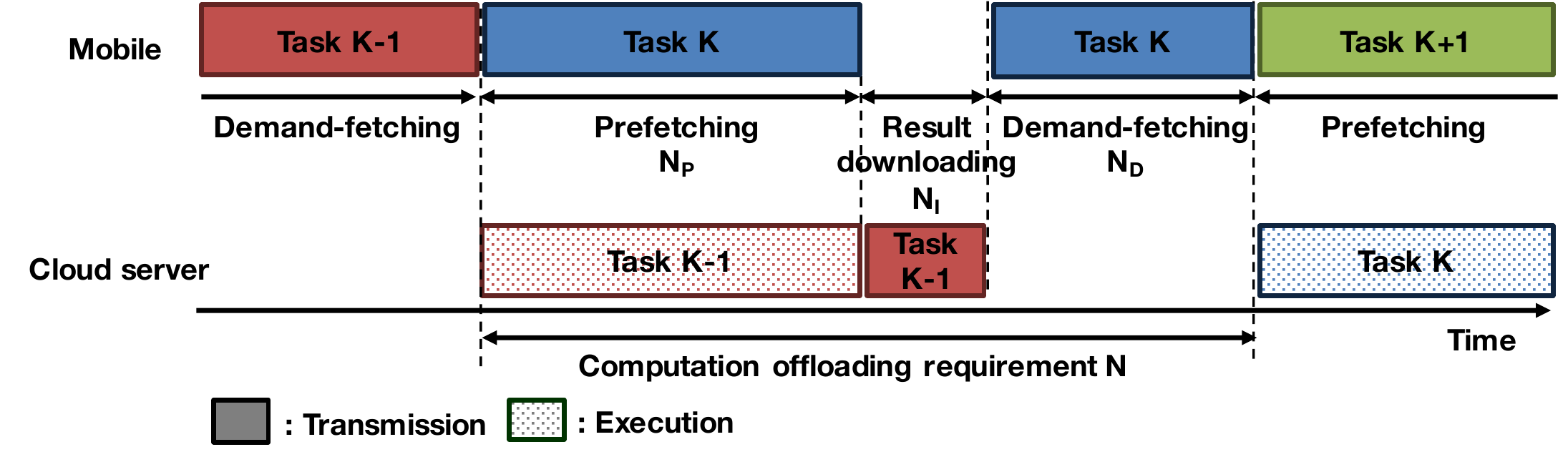}}\\
\subfigure[Conventional fetching architecture]{\includegraphics[width=3.5in]{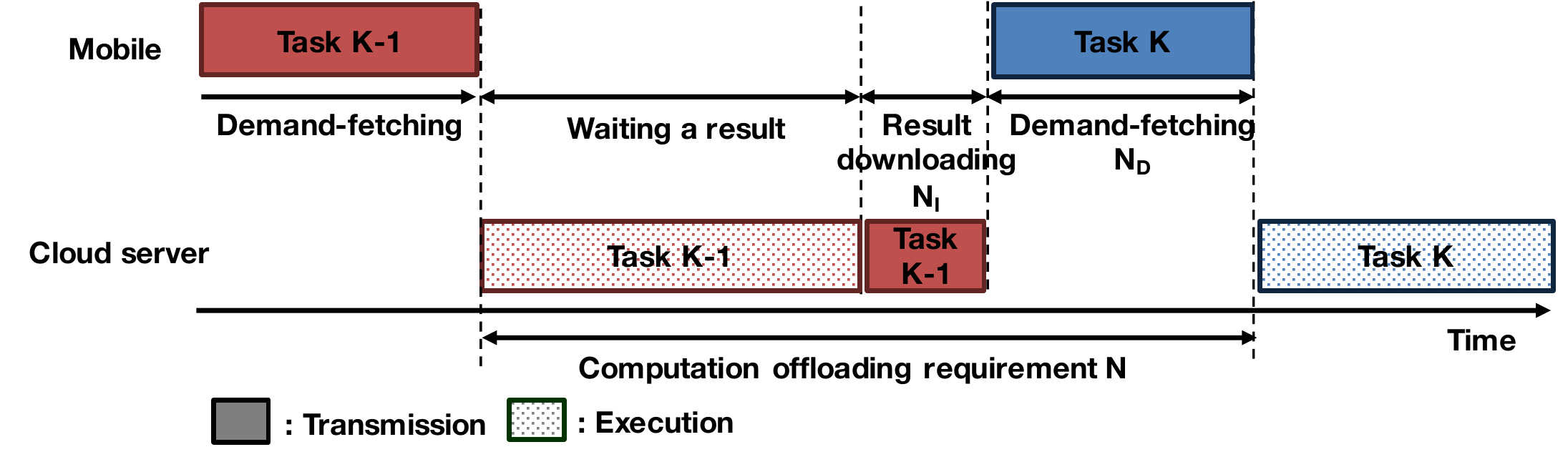}}\\
\caption{The MCO operation.  (a) The task topology  for the offloaded program.  (b)  MCO with live prefetching. (c) Conventional MCO with only demand fetching.}
 \label{Fig:BlockDiagram}
\end{figure}

As aforementioned, fetching refers to the process of delivering user-specific input data (e.g., images, videos, and measured sensor values) from a mobile to the cloud in advance of  computation. 
We define {\it task data} as  new input data with respect to the corresponding task.
Depending on the instant when TD is delivered to the cloud, two fetching schemes are considered defined as follows.
\begin{itemize}
\item {\it Prefetching}: 
A task is said to be prefetched when the mobile offloads part of its TD 
during the computation period of the previous task.
\item  {\it Demand-fetching}: 
A task is said to be demand-fetched
when the mobile offloads un-prefetched TD of the task 
after the completion of the preceding task. 
\end{itemize}

The proposed live prefetching and the conventional fetching architecture are compared in Fig.~\ref{Fig:BlockDiagram}. 
In the proposed architecture, the time overlap  between cloud computation and prefetching reduces latency. Let $(K-1)$ and $K$  denote the indices of the  current and next tasks, respectively. Consider the duration of simultaneous computation of Task $(K-1)$ and prefetching of Task $K$ as shown in Fig.~\ref{Fig:BlockDiagram}(b). Time is divided into slots. The duration is divided into three sequential phases: $N_P$ slots for prefetching of Task $K$ and computing Task $(K-1)$, $N_I$ for downloading computation results to the mobile, and $N_D$ for demand-fetching of Task $K$. The same operations repeat for subsequent tasks. In contrast, in the conventional  architecture without prefetching, fetching relies only on demand-fetching as shown in Fig.~\ref{Fig:BlockDiagram}(c). Comparing Fig.~\ref{Fig:BlockDiagram}(b) and 2(c), prefetching is observed to lengthen the  fetching duration from $N_D$ slots to $(N_P + N_D)$ slots, thereby reducing the transmission-energy consumption (see the system  model). Last, a latency constraint is applied such that the considered duration is no longer than $N$ slots and as a result, $N_{P}+N_{D}+N_I\leq N$. We assume that the result-download  period $N_I$ is negligible  given that  the access point, having much higher transmission  power than the mobile, enables delivering  the execution result within an extremely  short interval.  It follows that the demand-fetching period of $N_{D}$ slots is approximately equal to $(N-N_{P})$ slots. 

\begin{figure}[t]
\centering
\includegraphics[width=6cm]{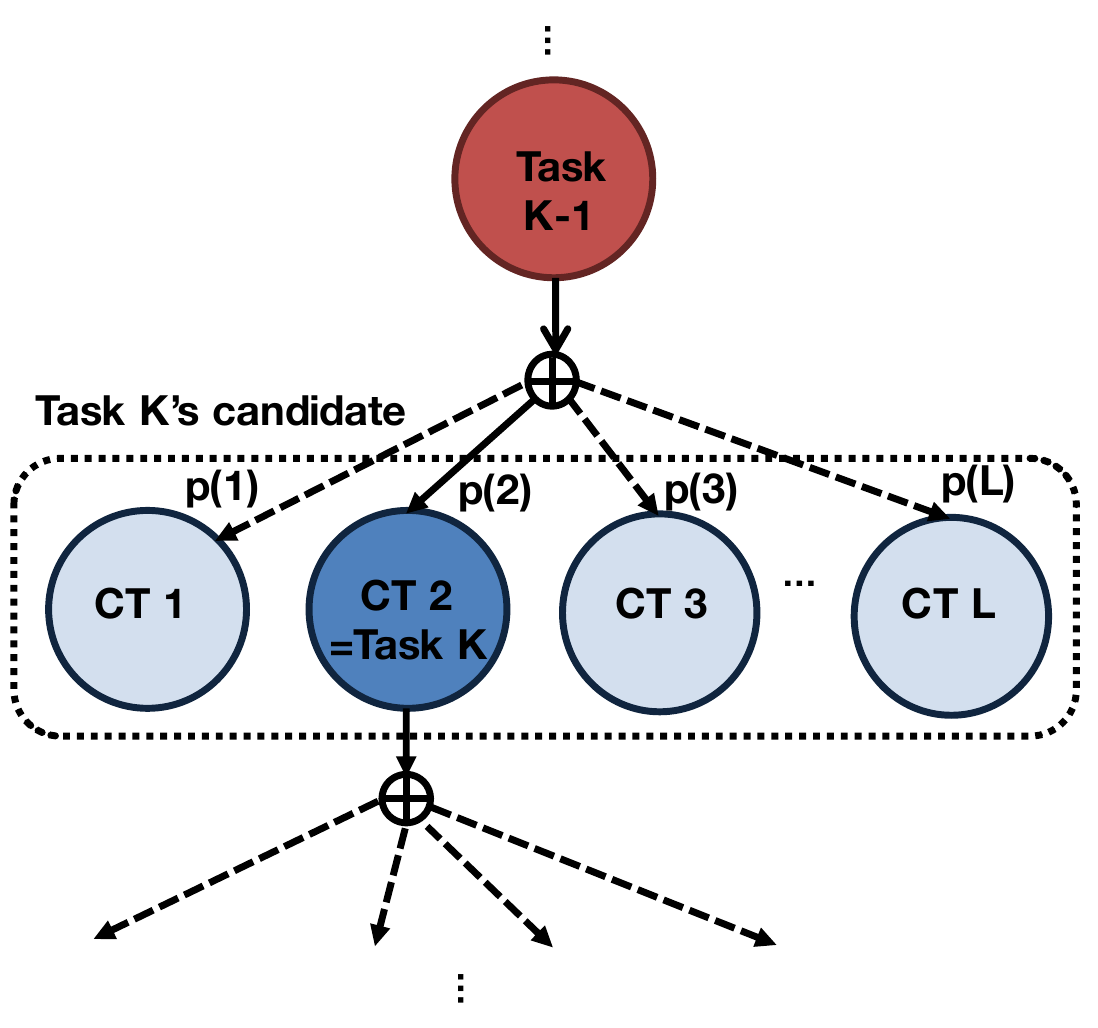}
\caption{Task execution graph.}
\label{Fig:FetchProcess}
\end{figure}

The live prefetching operation is elaborated as follows. 
The optimal prefetching controlled by the cloud  depends on the specified sequence of executed tasks. 
The sequence  is assumed to be unknown to the cloud in advance and modeled as a stochastic sequence 
 using the task-transition graph shown in Fig.~\ref{Fig:FetchProcess} following the approach in \cite{Chou1982}. 
The states of the graph correspond to potential tasks and the arrows are feasible sequential transitions from tasks to tasks 
of which the weights give the transition probabilities. 
Consider the current state being the execution of Task $(K-1)$ [see Fig.~\ref{Fig:FetchProcess}].  
There exist $L$ \emph{candidate tasks} (CTs) for Task $K$ with corresponding transition probabilities denoted as $p(1), p(2), \cdots, p(L)$ and $\sum_{\ell=1}^{L} p(\ell)=1$. Then among the CTs, one is selected as Task $K$ upon the completion of Task $(K-1)$. However, at the current state, the realization of Task $K$ is not yet known and thus not all prefetched TD can be useful\footnote{Noting that the input data for different tasks may overlap, a higher-order task transition graph enables to design 
more efficient architecture by reusing offloaded data before but it is too difficult to obtain the tractable structure. 
Instead, our work focuses on the one-task transition graph shown in Fig.~\ref{Fig:FetchProcess} for tractability.}.

The live prefetching policy is designed  in the sequel to maximize the utility of prefetched data. 
To this end, given the current state being Task $(K-1)$,  let $\gamma(\ell)$ denote the size of TD (in bits) 
if CT~$\ell$ is selected as Task~$K$. We assume that TD is for  exclusive use, namely that TD of CT~$\ell$ 
is useless for other CTs. 
The total TD satisfies the following constraint
\begin{align}\label{TDandPTD}
\sum_{\ell=1}^{L} \gamma(\ell)=\Gamma,
\end{align}
where the constant $\Gamma$ represents the sum of TDs for all CTs.  The TD for CT $\ell$  is divided into prefetching and demand-fetching parts, denoted as  $\alpha(\ell)$ and $\beta(\ell)$, respectively: 
\begin{align}\label{Notation:gammaalphabeta} 
\gamma(\ell)=\alpha(\ell)+\beta(\ell).
\end{align}
In the aforementioned duration of $N$ slots,  the mobile performs the following two-stage fetching for  MCO of Task~$K$:
\begin{enumerate}
\item \emph{Prefetching stage}:  During the $N_P$  slots, the mobile performs prefetching by transmitting $\sum_{\ell=1}^L \alpha(\ell)$~bits.
\item \emph{Demand-fetching stage}:  During the $N_D$  slots, the mobile performs demand-fetching by transmitting $\beta(\ell)$ bits, given that CT $\ell$ is chosen as Task $K$. 
\end{enumerate} 
\subsection{Problem Formulation}

Consider the $N$-slot duration of fetching for Task $K$ [see Fig~\ref{Fig:BlockDiagram}(b)]  or equivalently the current state of the fetching process in Fig.~\ref{Fig:FetchProcess} being Task $(K-1)$. Let $\alpha(\ell)$ denote the number of prefetched  bits for CT~$\ell$ with $1\leq \ell \leq L$ and define the prefetching vector  $\boldsymbol\alpha=\l[\alpha(1), \alpha(2), \cdots, \alpha(L)\r]$. Similarly, the demand-fetching vector is defined as $\boldsymbol\beta=\l[\beta(1), \beta(2), \cdots, \beta(L)\r]$ where $\boldsymbol{\alpha}+\boldsymbol\beta=\boldsymbol\gamma=[\gamma(1),\gamma(2),\cdots,\gamma(L)]$. Moreover, recall that $b_n$ represents the number of bits the mobile transmits to the cloud in slot $n$. Given $\boldsymbol\alpha$ and $\sum_{n=1}^{N_P} b_{n}=\sum_{\ell=1}^L \alpha(\ell)$,  the expected energy consumption of prefetching, represented by $E_{P}$,  is given as 
\begin{align}
E_{P}(\boldsymbol\alpha)= \mathbb{E}_g\left[\sum_{n=1}^{N_P}\mathcal{E}(b_{n},g_{n})\right] \nonumber
\end{align}
with the energy function  $\mathcal{E}$ given in \eqref{EnergyModel}. 
For demand-fetching, one of the CTs, say CT $\ell$,  is chosen as Task $K$. This results in the energy consumption of demand-fetching, denoted as $E_D$, given as 
\begin{align}\label{Eq:EnergyFun:DF}
E_{D}\l(\beta(\ell)\r)= \mathbb{E}_g\left[\sum_{n=N_P +1 }^{N}\mathcal{E}(b_{n},g_{n})\right]
\end{align}
with  $\sum_{n=N_P +1 }^{N} b_{n}=\beta(\ell)$.

The live prefetching design is formulated as a two-stage stochastic optimization problem under the criterion of maximum energy efficiency. 
In the first-stage, consider the design of optimal demand-fetcher conditioned on the prefetching decisions. 
Given the prefetching vector~$\boldsymbol\alpha$   and the realization of Task $K$ as CT $\ell$, 
the total number of bits for demand fetching is obtained as $\beta(\ell) = \gamma(\ell) - \alpha(\ell)$ based on \eqref{Notation:gammaalphabeta}.  
Then the design is translated into the problem of optimal adaptive transmission over $(N - N_P)$ slots for maximizing the energy efficiency, formulated as~follows: 
\begin{align}\label{DemandFetcherOptimization}\tag{P1}
\begin{aligned}
\underset{\{b_{N_P +1}, \cdots, b_N\}}{\min} &\quad  E_{D}(\beta(\ell)) \\  
\text{s.t.}&\quad   \sum_{n=N_P +1 }^{N} b_{n}=\beta(\ell),
\end{aligned}
\end{align}
where the objective function $E_{D}$ is given in \eqref{Eq:EnergyFun:DF}. 

Let $E_{D}^*(\beta(\ell))$ represent the solution for Problem P1. 
In the  second-stage, the prefetching policy is optimized to minimize the overall energy  consumption for fetching. The corresponding optimization problem is formulated as follows: 
\begin{align}\label{TwoStageOptimization}\tag{P2} 
\begin{aligned}
 \underset{\boldsymbol{\alpha},\ \boldsymbol{\beta}, \{b_1, \cdots, b_{N_P}\}}{\min} & \quad E_{P}(\boldsymbol{\alpha})+\sum_{\ell = 1}^L {E}^*_{D}(\beta(\ell))p(\ell)\\
 \text{s.t.}  &\quad \boldsymbol{\alpha}+ \boldsymbol{\beta}=\boldsymbol{\gamma},\\
 &\quad \boldsymbol{\alpha}, \boldsymbol{\beta}\geq \boldsymbol{0}.
\end{aligned}
\end{align}
The optimal live prefetching policy is designed in the sequel by solving Problems \ref{DemandFetcherOptimization} and \ref{TwoStageOptimization}.

\section{Live Prefetching over Slow Fading Channels}\label{Section:SlowFadingPart}

This section aims at designing the optimal live prefetching policy for slow fading channels 
where channel gains remain constant over the fetching period of $N$ slots. 
We derive the mathematical expression of the optimal prefetching vector $\boldsymbol\alpha^*$ and 
quantify the prefetching gain. 

\subsection{Optimal Live Prefetching Policy for Slow Fading}

To facilitate the policy derivation, Problem  \ref{TwoStageOptimization} is transformed as follows. First, it requires solving Problem P1. To this end, the energy consumption of demand-fetching $E_D$ in  \eqref{Eq:EnergyFun:DF} for slow fading can be written and bounded  as 
\begin{align}
E_{D}\l(\beta(\ell)\r)&=\frac{\lambda}{g} \cdot \sum_{n=N_P+1}^N (b_n)^m\nonumber\\
&\geq 
\frac{\lambda}{g} \cdot (N-N_P) \cdot({b_{N_P+1}\times \cdots \times b_{N}})^{\frac{m}{N-N_P}},\nonumber
\end{align}
where the lower bound  follows from the inequality of arithmetic and geometric means. The equality holds when $\{b_n\}$ are equal, 
and the  resultant energy consumption for demand fetching~is 
\begin{align}
E_{D}^*\l(\beta(\ell)\r)=\frac{\lambda}{g} \cdot \frac{\beta(\ell)^m}{(N-N_P)^{m-1}}.\nonumber
\end{align}
The objective function of Problem \ref{TwoStageOptimization} thus becomes
\begin{align}\label{Lemma1:Proof1}
&\frac{\lambda}{g}\cdot \sum_{n=1}^{N_P} (b_n)^m+\frac{\lambda}{g}\cdot \frac{\beta(\ell)^m}{(N-N_P)^{m-1}}\nonumber\\
\geq 
&\frac{\lambda}{g} \cdot ({b_{1}\cdots b_{N_P}})^{\frac{m}{N_P}}+\frac{\lambda}{g} \cdot \frac{\beta(\ell)^m}{(N-N_P)^{m-1}},
\end{align}
where the lower bound follows again from  the inequality of arithmetic and geometric means. 
Recalling that $\sum_{n=1}^{N_P}b_n=\sum_{\ell=1}^L \alpha(\ell)$, the equality holds if
$b_1=\cdots=b_{N_P}=\frac{\sum_{\ell=1}^L \alpha(\ell)}{N_P}$. Substituting this result  into \eqref{Lemma1:Proof1} yields the following lemma.

\begin{lemma} \label{Lemma:ProblemFormulation:StaticChannel}
\emph{Given slow fading, Problem \ref{TwoStageOptimization} is equivalent to:
\begin{align}\label{TwoStageOptimization:SlowFading}\tag{P3} 
\begin{aligned}
& \underset{\boldsymbol{\alpha}}{\min} & & \frac{\l(\sum_{\ell=1}^{L} \alpha(\ell)\r)^m}{(N_P)^{m-1}}+\frac{\sum_{\ell=1}^L p(\ell)(\gamma(\ell)-\alpha(\ell))^m}{(N-N_P)^{m-1}}\\
& \text{s.t.}  && 0 \leq \boldsymbol{\alpha}\leq \boldsymbol{\gamma},
\end{aligned}
\end{align}
and the number of offloaded bits $b_n$ is 
\begin{align}
b_n=\l\{
\begin{aligned}
& \frac{\sum_{\ell=1}^L \alpha(\ell)}{N_P},  && n=1,\cdots, N_P,\\
& \frac{\gamma(\ell)-\alpha(\ell)}{N-N_P}, && n=N_P+1,\cdots, N.
\end{aligned}
\r.\nonumber
\end{align}  
}
\end{lemma}
Lemma~\ref{Lemma:ProblemFormulation:StaticChannel} shows 
that the offloaded bits $b_n$ are evenly distributed over multiple slots, 
removing the variable $\boldsymbol\beta$ in  \ref{TwoStageOptimization} and thereby allowing tractable policy analysis.
The main result is shown in the following proposition. 
\begin{proposition}\label{Proposition:OptimalPrefetchingPolicy:SF}
\emph{ (Optimal Prefetching for Slow Fading) Given slow fading, 
the optimal prefetching vector $\boldsymbol\alpha^*=[\alpha^*(1),\cdots, \alpha^*(L)]$ is 
\begin{align}\label{Eq:Optimalalpha:SF}
{\boldsymbol\alpha}^*=\l[\boldsymbol\gamma-\boldsymbol{p}^{-\frac{1}{m-1}}\frac{N-N_P}{N_P}\alpha^*_{\Sigma}\r]^+,
\end{align}
where  $\boldsymbol{p}=[p(1),p(2),\cdots,p(L)]$ is the task transition probability vector and $\alpha^*_{\Sigma}$ is the optimal total number of prefetched bits,
$\alpha^*_{\Sigma}=\sum_{\ell=1}^L \alpha^*(\ell)$.}
\end{proposition}
\begin{proof}
See Appendix~\ref{App:OptimalPrefetchingPolicy:SF}. 
\end{proof}

The  optimal prefetching vector $\boldsymbol\alpha^*$ in \eqref{Eq:Optimalalpha:SF} determines the prefetching-task set $\mathcal{S}$ defined in Definition~\ref{definition:PTS}  and vice versa. Exploiting this relation, $\boldsymbol\alpha^*$ and $\mathcal{S}$ can be computed using the simple iterative algorithm presented in Algorithm~\ref{Algo:SlowFade} where the needed \emph{prefetching-priority function} is defined in Definition~\ref{definition:PPF}.  The existence of $\boldsymbol\alpha^*$ (or equivalently $\mathcal{S}$) is shown in Corollary~\ref{Corollary:UniqueOptimalPrefetchingVector:SF}. In addition, given $\boldsymbol\alpha^*$ and $\mathcal{S}$, the optimal total prefetched bits $\alpha^*_{\Sigma}$ is given as 
\begin{align}\label{Eq:TotalPrefetchedBits:SF}
\alpha^*_{\Sigma}=\frac{\sum_{\ell \in\mathcal{S}} \gamma(\ell)}{1+\frac{N-N_P}{N_P}\sum_{\ell \in\mathcal{S}}p(\ell)^{-\frac{1}{m-1}}}. 
\end{align}

\begin{definition}\label{definition:PTS}\emph{(Prefetching Task Set)  The  {\it prefetching-task set}, denoted as $\mathcal{S}$,   is defined as the set of prefetched CTs:    $\mathcal{S}=\left\{\ell \in \mathbb{N} \left\lvert \alpha^*(\ell)>0, 0\leq \ell \leq L \right. \right\}$, where the needed prefetching priority function $\delta(\cdot)$ is defined in  Definition~\ref{definition:PPF}.
}
\end{definition}
\begin{definition}\label{definition:PPF}\emph{(Prefetching Priority Function) The  {\it prefetching priority function} of CT~$\ell$ is defined as $\delta(\ell)=\gamma(\ell) p(\ell)^{\frac{1}{m-1}}$, which determines the prefetching order (see Algorithm~\ref{Algo:SlowFade}). }
\end{definition}

\begin{corollary} 
\label{Corollary:UniqueOptimalPrefetchingVector:SF}
 \emph{(Existence of the Optimal Prefetching Vector $\boldsymbol\alpha^*$)
A unique optimal prefetching vector $\boldsymbol\alpha^*$ satisying \eqref{Eq:Optimalalpha:SF}  always exists in the range of
$0\preceq  \boldsymbol\alpha^*\preceq \boldsymbol\gamma$\footnote{The symbols $\preceq$ in Corollary \ref{Corollary:UniqueOptimalPrefetchingVector:SF} and $\succ$ in Algorithms \ref{Algorithm:FindingPS:SF} and \ref{Algorithm:FindingPS:FF} denote the element-wise inequalities.}
where the first and second equalities hold only when $N\rightarrow \infty$ or $N=N_P$, respectively.  
}\end{corollary}
\begin{proof}
See Appendix~\ref{App:UniqueOptimalPrefetchingVector:SF}. 
\end{proof}

\begin{remark} \emph{(Prefetched or Not?) Corollary~\ref{Corollary:UniqueOptimalPrefetchingVector:SF} shows 
that the optimal prefetching vector $\boldsymbol\alpha^*$ is strictly positive if $N$ is finite. 
In other words, prefetching in a slow fading channel is always beneficial to a computation task requiring a finite latency constraint. 
On the other hand, for a latency-tolerant task ($N\rightarrow \infty$),  a mobile needs not perform prefetching ($\boldsymbol\alpha^*=0$).     
}
\end{remark}

\begin{remark} \emph{(Partial or Full Prefetching?) Corollary~\ref{Corollary:UniqueOptimalPrefetchingVector:SF} shows 
that if $N>N_P$, the optimal prefetching vector $\boldsymbol\alpha^*$ is strictly  less than $\boldsymbol\gamma$ 
because the remaining number of bits $\beta(\ell)$ can be delivered during the demand-fetching duration $(N-N_P)$. 
If $N=N_P$, on the other hand, full prefetching ($\boldsymbol\alpha^*=\boldsymbol\gamma$) is optimal because only prefetching is possible.   
}
\end{remark}

\begin{algorithm}[t]
\caption{Finding the optimal prefetching vector and prefetching-task set for slow fading}\label{Algorithm:FindingPS:SF}
\begin{algorithmic}[1]
\State Arranging the CTs in a descending order in terms of the prefetching-priority function $\delta$ in Definition~\ref{definition:PPF}, e.g. $\delta(\ell_1)\geq \delta(\ell_2)$, if and only if $\ell_1>\ell_2$.
\State Setting $\ell=0$ and $\mathcal{S}=\emptyset.$ 
\While{$|{\mathcal{S}}|<L$}
   \State $\ell = \ell+1$	and $\mathcal{S} = \mathcal{S}\bigcup \{\ell\}$.
   \State Compute   $\alpha^*_{\Sigma}$ using \eqref{Eq:TotalPrefetchedBits:SF}.
   \State Compute  $\boldsymbol\alpha^*$ using \eqref{Eq:Optimalalpha:SF}.
   \State Count the number of positive elements in $\boldsymbol\alpha^*$, 
   namely  $|\boldsymbol\alpha^* \succ \boldsymbol 0|$. 
   \If{$|\boldsymbol\alpha^* \succ \boldsymbol 0|=|\mathcal{S}|$ }
   \State \textbf{break}
   \EndIf 
\EndWhile
  \State \textbf{return} $\boldsymbol\alpha^*$ and $\mathcal{S}$.
\end{algorithmic}\label{Algo:SlowFade}
\end{algorithm}

\subsection{Prefetching Gain for Slow Fading}
In order to quantify how much prefetching increases the energy efficiency, we introduce and define the prefetching gain as follows. 
\begin{definition}\label{Def:PrefGain} \emph{ (Prefetching Gain) A prefetching gain $G_P$ is defined as the energy-consumption ratio between MCOs without and with prefetching,
\begin{align}\label{Eq:PrefetchingGainDefinition}
G_P=\frac{\sum_{\ell = 1}^L {E}^*_{D}(\gamma(\ell))p(\ell)}{E_P(\boldsymbol\alpha^*)+\sum_{\ell\in \mathcal{S}} p(\ell) E^*_D(\beta(\ell))}.
\end{align} 
}
\end{definition}
 
The prefetching gain $G_P$ depends on several factors including the latency requirement $N$, prefetching duration $N_P$, and the number of CTs $L$. 
The following result specifies  the relationship mathematically. 
\begin{proposition} \label{Theorem:PrefetchingGain:SF} 
\emph{(Prefetching Gain for Slow Fading).  
Given slow fading, the prefetching gain is 
\begin{align}\label{Theorem:LowerPrefetchingGain}
G_P \geq \l[\frac{N-N_P (1-L^{-\frac{m}{m-1}})}{N-N_P }\r]^{m-1},
\end{align}
where the equality holds when $\gamma(\ell)=\frac{\Gamma}{L}$ and $p(\ell)=\frac{1}{L}$ for all $\ell$.
}
\end{proposition}
\begin{IEEEproof}
See Appendix~\ref{App:PrefetchingGain:SF}. 
\end{IEEEproof}
It can be observed that the prefetching gain 
$G_P$ is strictly larger than one.  This shows the effectiveness of optimally controlled prefetching. 

\begin{remark} \emph{ (Effects of Prefetching Parameters) 
The prefetching gain $G_P$ in \eqref{Theorem:LowerPrefetchingGain} shows 
the effects of different parameters. On one hand, 
prefetching lengthens the fetching duration from $(N-N_P)$ to $N$ slots, 
reducing  the transmission rate and thus the resultant energy consumption.  
On the other hand, the mobile pays for the cost of prefetching $N_P(1-L^{-\frac{m}{m-1}})$ due to transmitting   redundant bits. The gain is larger than the cost and hence prefetching is beneficial. 
}
\end{remark}

\section{Live Prefetching over Fast Fading Channels}\label{Section:FastFadingPart}

In this section, we derive the optimal prefetching policy that solves Problem \ref{TwoStageOptimization} 
for fast fading where channel gains are i.i.d. over different slots.
Unlike the case of slow fading, the numbers of bits allocated over time slots
 $\{b_n\}$ should be 
jointly optimized with the prefetching vector $\boldsymbol\alpha$ 
because $b_n$ depends on not only the channel gain $g_n$ but also previous allocated bits $b_1,\cdots\ b_{n-1}$.   
To this end, we reformulate Problems \ref{DemandFetcherOptimization} and \ref{TwoStageOptimization} 
to allow sequential optimization, and solve them by using the backward policy iteration  from stochastic optimization. 

\subsection{Optimal Demand-Fetching Policy for Fast Fading}
The optimal demand-fetching policy is derived by solving Problem P1 as follows.  The demand-fetcher allows the mobile to transmit  $\beta(\ell)=\gamma(\ell)-\alpha(\ell)$  bits 
from slot $(N_P+1)$ to $N$.
Designing the optimal demand-fetcher in Problem \ref{DemandFetcherOptimization} for fast fading 
is thus equivalent to finding the optimal \emph{causal scheduler} \cite{Lee2009} to deliver $\beta(
\ell)$ bits within $(N-N_P)$ slots.
Using the approach of dynamic programming, Problem \ref{DemandFetcherOptimization} is rewritten as
\begin{align} \label{OptimalBit}\tag{P4}
&J_{n}(\rho_{n}(\ell),g_{n})=\nonumber\\
&\l\{
\begin{aligned}
&\!\!\min_{0\leq b_{n}\leq \rho_{n}(\ell)}\!\! \left[\lambda \frac{\l(b_{n}\r)^{m}}{g_{n}}\!\!+\!\!\bar{J}_{n+1}(\rho_{n}(\ell)-b_{n}) \right], && \!\!\!\!N_P+1 \leq n< N,\\
&\!\!\lambda \frac{\l[\rho_{N_P}(\ell)\r]^{m}}{g_{N_P}}, && \!\!\!\!n=N,
\end{aligned}
\r.\nonumber
\end{align} where $\rho_{n}(\ell)$ is the number of remaining bits of CT $\ell$'s TD in slot $n$. 
The cost-to-go function~$\bar{J}_{n}(\rho_{n}(\ell))=\mathbb{E}_{g_{n}}[J_{n}(\rho_{n}(\ell), g_{n})]$ of Problem \ref{OptimalBit} 
gives the expected energy consumption for transmitting $\rho_n(\ell)$ bits from slot $n$ to $N$ 
if the optimal demand-fetching policy is utilized in the remaining slots.  Solving Problem P4 using the technique of backward policy iteration gives the following main result. 
\begin{proposition}\label{Proposition:OptimalDemandFetcher:FF} \emph{ (Optimal Demand-Fetching for Fast Fading)
Given fast fading and the selected task, Task $K$, being CT~$\ell$, the optimal bit allocation over the demand-fetching slots is 
\begin{align}\label{Proposition:OptimalDemandFetcher}
b_n^*(\rho_n(\ell), g_n)=\frac{\rho_{n}(\ell) (g_{n})^{\frac{1}{m-1}}}{(g_{n})^{\frac{1}{m-1}}+ \l(\frac{1}{\xi_{N-n}}\r)^{\frac{1}{m-1}}}, 
\end{align}
where the sequence $\{\rho_{n}(\ell)\}$ can be computed using the initial value $\rho_{N_P+1}(\ell)= \beta(\ell)$ and computed recursively as  $\rho_{n}(\ell)=\rho_{n-1}(\ell)-{b}^*_{n-1}$,  and the coefficients $\{\xi_n\}$ are defined as 
\begin{align}\label{Proposition:xi}
\!\!\!\!\xi_{N-n}\!\!=\!\!
\l\{
\begin{aligned}
&\! \mathbb{E}_{g}\!\!\l[\tfrac{1}{\l(g^{\frac{1}{m-1}}\!+\!\l(\frac{1}{\xi_{N-n-1}}\r)^{\frac{1}{m-1}}\!\!\r)^{{m-1}}}\r], 
&& \!\!\!\!\textrm{$n<N$, }\\
& \! \infty, 
&& \!\!\!\!\textrm{$n=N$. }
\end{aligned}
\r.
\end{align}
The corresponding  expected energy consumption for demand-fetching is 
\begin{align}\label{Proposition:MinimumEnergyForDF}
{E}_{D}^*\l(\beta(\ell)\r)=\lambda  \xi_{N-N_P} {\beta(\ell)}^m.
\end{align}
}
\end{proposition}
\begin{proof}
See Appendix \ref{App:OptimalDemandFetcher:FF}. 
\end{proof}
In addition, bounds on the demand-fetching energy consumption can be obtained as shown in  Corollary~\ref{Corollary:DemandFetchingEnergyBound:FF}, which are useful  for designing the sub-optimal prefetching policies. 
\begin{corollary}\label{Corollary:DemandFetchingEnergyBound:FF}\emph{(Demand-Fetching Energy Consumption for Fast Fading)  
The minimum energy consumption ${E}_{D}^*\l(\beta(\ell)\r)$ in  \eqref{Proposition:MinimumEnergyForDF} 
can be lower and upper bounded as
\begin{align}\label{Eq:DemandFetchingEnergyBound:FF}
\frac{\lambda \beta(\ell)^m}{\mathbb{E}[g] (N-N_P)^{m-1}}  
\leq {E}_{D}^*\l(\beta(\ell)\r)
\leq \mathbb{E}\l[\frac{1}{g}\r]\frac{\lambda \beta(\ell)^m}{ (N-N_P)^{m-1}},
\end{align}
where the first and second equalities hold when $N-N_P\rightarrow \infty$ and $N-N_P=1$, respectively.
} 
\end{corollary}
\begin{proof}
See Appendix \ref{App:DemandFetchingEnergyBound:FF}.
\end{proof}

\subsection{Optimal Prefetching for Fast Fading}
It is intractable to drive  the optimal prefetching policy in closed form  by directly solving Problem P2 for the case of fast fading. The main difficulty lies in deriving the optimal prefetching-task set $\mathcal{S}$ in Definition~\ref{definition:PTS} given CSI causaliy. Assuming it is known, the optimal prefetching policy can be obtained in closed form that reveals  a threshold based structure similar to the slow-fading counterpart in Proposition~\ref{Proposition:OptimalPrefetchingPolicy:SF}, which is useful for designing sub-optimal policies in the sequel. To this end, some notations are introduced. Let $s_{n}(\ell)$ represent the number of prefetched bits for CT $\ell$ in slot $n$   
satisfying $\sum_{\ell=1}^L s_{n}(\ell)=b_{n}$. As in the preceding sub-section, let $\rho_{n}(\ell)$ represent the remaining bits of CT $\ell$  to fetch at the beginning of slot $n$. Note that the remaining bits of CT~$\ell$ after  prefetching, ${\rho}_{N_P+1}(\ell) $,  is equal to $\beta(\ell)$ in \eqref{Notation:gammaalphabeta}. To simplify notation, the variables defined above are grouped as the \emph{decision vector}, $\boldsymbol{s}_{n}=\l[s_{n}(1), s_{n}(2), \cdots, s_{n}(L)\r]$, and the \emph{status vector},  $\boldsymbol{\rho}_{n}=[\rho_{n}(1), \cdots, \rho_{n}(\ell),\cdots, \rho_{n}(L)]$ with $n=1, \cdots, N_P$. Note that  $\sum_{n=1}^{N_P} \boldsymbol{s}_{n}=\boldsymbol\alpha$. Moreover, the initial value $\boldsymbol{\rho}_{1}$ is $\boldsymbol\gamma$, 
and $\boldsymbol{\rho}_{n}$ can be calculated recursively as 
$\boldsymbol{\rho}_{n}=\boldsymbol{\rho}_{n-1}-\boldsymbol{s}_{n-1}$.

The  optimal  decision vector $\boldsymbol{s}_{n}^*$ depends on the energy consumption not only in  
 the current slot but also in the remaining prefetching and demand-fetching periods. This suggests that Problem P2 for fast fading can be reformulated as a sequential decision problem. Using $\sum_{\ell=1}^L s_{n}(\ell)=b_{n}$, the energy consumption in prefetching slot $n$  is thus given as
\begin{align}
\mathcal{E}_{n}(b_{n},g_{n})=\lambda \frac{\l(\sum_{\ell=1}^L s_{n}(\ell)\r)^m}{g_{n}}.\nonumber
\end{align}
Then Problem P2 can be translated into the sequential decision problem as shown in Problem \ref{OptimizationPrefetcher} as
{\small
\begin{align}\label{OptimizationPrefetcher}\tag{P5}
&V_{n}(\boldsymbol{\rho}_{n},g_{n})\!=\!\nonumber\\
&\l\{
\begin{aligned}
&\min_{0\leq \boldsymbol{s}_{n}\leq \boldsymbol{\rho}_{n}} \left[\lambda \frac{\l(\sum_{\ell=1}^L s_{n}(\ell)\r)^{m}}{g_n}
+\bar{V}_{n+1}(\boldsymbol{\rho}_{n}-\boldsymbol{s}_{n})\right], &&  {n}<N_P,\\
& \min_{0\leq \boldsymbol{s}_{n}\leq \boldsymbol{\rho}_{n}} \left[\lambda \frac{\l(\sum_{\ell=1}^L s_{n}(\ell)\r)^{m}}{g_{n}}+\lambda\cdot \xi_{N-N_P}\right.\\
&\left. \quad\quad\quad\quad\quad\quad\quad\quad\quad  \cdot \sum_{\ell=1}^L p(\ell) (\rho_{n}(\ell)-s_{n}(\ell))^m \right],  && {n}=N_P,
\end{aligned}
\r.\nonumber
\end{align}}where $\bar{V}_{n}(\boldsymbol{\rho})=\mathbb{E}_{g}\l[V_{n}(\boldsymbol{\rho},g)\r]$ is the cost-to-go function
measuring the cumulative expected energy consumption from  slot $n$ to $N$   if the optimal prefetching and demand-fetching are performed in these slots. 
Note that the term $\lambda \xi_{N-N_P}\sum_{\ell=1}^L p(\ell) (\rho_{N_P}(\ell)-s_{N_P}(\ell))^m$ is the expected energy consumption of the optimal demand-fetching in Proposition~\ref{Proposition:OptimalDemandFetcher:FF} given the  decision vector $\boldsymbol{s}_{N_P}$. The optimal prefetching policy can be computed numerically using the technique of backward policy iteration but the result yields  little insight into the policy structure. The difficulty in deriving  the policy analytically aries from the CSI causality. Specifically, deriving $\boldsymbol{s}_{n}^{*}$ for $n<N_P$ requires the cost-to-go function  $\bar{V}_{n}$ to have a closed-form expression which does not exist since the function depends on future channel realizations, $g_{n+1},\cdots, g_{N_P}$, which are unknown in slot $n$.

Nevertheless, useful insight into the optimal policy structure can be derived if the CSI causality constraint can be relaxed such that  the prefetching-task set $\mathcal{S}$  can be computed. This leads to the following main result. 
\begin{proposition}\label{Proposition:OptimalPrefetcher:FF} \emph{(Optimal  Prefetching for Fast Fading) Assuming that the prefetching-task set $\mathcal{S}$ in Definition~\ref{definition:PTS} is known, 
the optimal decision vector, $\boldsymbol{s}_{n}^*=[s_{n}^*(1), \cdots, s_{n}^*(L)]$,    which solves Problem \ref{OptimizationPrefetcher}, is  
\begin{align}\label{Proposition:OptmalPrefetcher1}
\boldsymbol{s}_{n}^*=\l[\boldsymbol{\rho}_{n}-\eta_{n} \boldsymbol{p}^{-\frac{1}{m-1}}\r]^+,
\end{align}
where the thresholds $\{\eta_{n}\}$ are defined as
{\small
\begin{align}\label{Proposition:OptmalPrefetcher2}
&\eta_{n}=\nonumber\\
&\l\{
\begin{aligned}
&\frac{\sum_{\ell\in \mathcal{S}}\rho_{n}(\ell)\l(\frac{1}{\zeta_{N-n}(\mathcal{S})}\r)^{\frac{1}{m-1}}}
{\l[g_{n}^{\frac{1}{m-1}}+\l(\frac{1}{\zeta_{N-n}(\mathcal{S})}\r)^{\frac{1}{m-1}}\r]\sum_{\ell \in \mathcal{S}}
p(\ell)^{-\frac{1}{m-1}}},
&& \textrm{$n<N_P$,}\\
&\frac{\sum_{\ell \in \mathcal{S}}\rho_{n}(\ell)\l(\frac{1}{\xi_{N-N_P}}\r)^{\frac{1}{m-1}}}{g_{n}^{\frac{1}{m-1}}+\l(\frac{1}{\xi_{N-N_P}}\r)^{\frac{1}{m-1}}
\sum_{\ell \in \mathcal{S}}p(\ell)^{-\frac{1}{m-1}}},
&& \textrm{$n=N_P$,}
\end{aligned}
\r.
\end{align}}with the coefficients $\{\zeta_{N-n}(\mathcal{S})\}$ therein defined as
{\small
\begin{align}\label{Proposition:OptmalPrefetcher3}
&\zeta_{N-n}(\mathcal{S})\!=\!\nonumber\\
&\l\{
\begin{aligned}
&\mathbb{E}_{g}\!\!\l[\l(g^{\frac{1}{m-1}}+\l(\tfrac{1}{\zeta_{N-n-1}(\mathcal{S})}\r)^{\frac{1}{m-1}}\r)^{{-(m-1)}}\r], 
&&\!\! \textrm{$n<N_P-1$,}\\
&\mathbb{E}_{g}\!\!\l[\l({g}^{\frac{1}{m-1}}+\l(\tfrac{1}{\xi_{N-N_P}}\r)^{\frac{1}{m-1}}\right. \right.\\
&\left.\left.\quad\quad\quad\quad\quad\quad \cdot \sum_{\ell \in \mathcal{S}} p(\ell)^{-\frac{1}{m-1}}
\r)^{-(m-1)}\r],  
&&\!\! \textrm{$n=N_P-1$.}
\end{aligned}
\r.
\end{align}
} and $\{\xi_{N-n}\}$ defined in \eqref{Proposition:xi}.
As a result, the total  expected energy consumption for fetching is
\begin{align}\label{Eq:TotalEnergyConsumption:FF}
&E_P(\boldsymbol\alpha^*)+\sum_{\ell=1}^L p(\ell) E^*_D(\beta(\ell))\nonumber\\
=&\lambda \l(\sum_{\ell\in \mathcal{S}}^L \gamma(\ell)\r)^m \zeta_{N}(\mathcal{S})+\lambda \sum_{\ell\notin \mathcal{S}} p(\ell) \gamma(\ell)^m,
\end{align}
where the first and second terms represent the energy consumption for prefetched and not-prefetched CTs, respectively.  
}
\end{proposition}
\begin{IEEEproof}
See Appendix~\ref{App:OptimalPrefetcher:FF}. 
\end{IEEEproof}

Comparing Propositions~\ref{Proposition:OptimalPrefetchingPolicy:SF} and \ref{Proposition:OptimalPrefetcher:FF}, both the optimal prefetching policies for slow and fast fading have a threshold based structure.  For the current case, the thresholds  $\{\eta_{n}\}$ in \eqref{Proposition:OptmalPrefetcher2} determine  which CTs are prefetched, and how many bits are prefetched for  each of these CTs.  In particular, a CT having  the prefetching priority  $\delta$ (see Definition~\ref{definition:PPF})  larger than  $\eta_n$ is prefetched in slot $n$.   Some basic properties for the thresholds are given as follows. 

\begin{corollary}\label{Corollary:PropertyThreshold} \emph{(Properties of Thresholds $\{\eta_{n}\}$)
The threshold $\eta_{n}$ in \eqref{Proposition:OptmalPrefetcher2} has the following properties. 
\begin{itemize}
\item (Monotone decreasing $\eta_n$) The sequence $\eta_{1}, \eta_2,\cdots,\eta_{N_P}$  is strictly monotone decreasing.
\item (Upper limit of $\eta_n$) The largest  threshold, namely  $\eta_1$,  is strictly less than the maximum prefetching priority: 
$\eta_1<\max_{\ell} \delta(\ell)$. 
\end{itemize}
}
\end{corollary}
\begin{proof}
See Appendix \ref{App:PropertyThreshold}.  
\end{proof}

Based on  earlier discussion, the second property in the corollary shows that the optimal prefetching should be always performed. 
\subsection{Sub-Optimal Prefetching for Fast Fading}
Proposition~\ref{Proposition:OptimalPrefetcher:FF} shows that the optimal prefetching policy can be obtained by conditioning on the knowledge of  the prefetching-task set $\mathcal{S}$. Unfortunately, finding $\mathcal{S}$ is impossible as it requires non-causal CSI, namely the chain gains $g_{n+1},\cdots, g_{N_P}$. In this section, two algorithms are designed for approximating $\mathcal{S}$ in every time slot. Let  $\widetilde{\mathcal{S}}_n$ denote the approximation result in slot $n$ with $n = 1, 2, \cdots, N_P$. 

First, we derive the procedure  for computing $\mathcal{S}$ and show its dependence on non-causal CSI. The results are useful for designing the said approximation algorithms. Consider slot~$n$ and assume that a candidate task, say  $\ell$, belongs to  $\mathcal{S}$. It follows that the corresponding optimal prefetched data size in the current and remaining slots are nonzero: $s_m^*(\ell)>0$ for $m = 1, 2, \cdots, N_P$. It is straightforward to obtain from   Proposition~\ref{Proposition:OptimalDemandFetcher:FF} the following recursive relation: $\rho_{n+1}(l)=
\rho_{n}(\ell)-s_{n}^*(\ell)=\eta_{n}(\ell)p(\ell)^{-\frac{1}{m-1}}$. Since the total prefetched bits for candidate task $\ell$ is $\alpha^*(\ell)=\sum_{n=1}^{N_P}{s}_{n}^*(\ell)$, it follows that 
 \begin{align}\label{Eq:SumofOptimalDecisionVector:FF}
\alpha^*(\ell)
=\l[\gamma(\ell)-\eta_{N_P} p(\ell)^{-\frac{1}{m-1}} \r]^+.
\end{align}
One can observe from Definition~\ref{definition:PTS} that $\{\alpha^*(\ell)\}$ and $\mathcal{S}$ can be computed iteratively similarly as in Algorithm~\ref{Algo:SlowFade}. Nevertheless, this is infeasible since the  the threshold $\eta_{N_P}$ needed for computing $\{\alpha^*(\ell)\}$ using \eqref{Eq:SumofOptimalDecisionVector:FF} depends on 
non-causal CSI according to the following lemma. 
\begin{lemma} \label{Lemma:Threshold}
\emph{ Considering  slot $n$ and given non-causal CSI  $g_n, g_{n+1},\cdots, g_{N_P}$ and the prefetching-task set $\mathcal{S}$,  the threshold $\eta_{N_P}$ is given as 
{\small
\begin{align}\label{Eq:FinalThreshold:FF}
\!\!\eta_{N_P}&=\frac{\sum_{\ell \in \mathcal{S}}\rho_n(\ell)\l(\frac{1}{\xi_{N-N_P}}\r)^{\frac{1}{m-1}}}{g_{N_P}^{\frac{1}{m-1}}+\l(\frac{1}{\xi_{N-N_P}}\r)^{\frac{1}{m-1}}\!\!\sum_{\ell \in \mathcal{S}}p(\ell)^{-\frac{1}{m-1}}}\nonumber\\
&\cdot \prod_{k=n}^{N_P-1}
\frac{ \l(\frac{1}{\zeta_{N-k}(\mathcal{S})}\r)^{\frac{1}{m-1}}}
{(g_k)^{\frac{1}{m-1}}+\l(\frac{1}{\zeta_{N-k}(\mathcal{S})}\r)^{\frac{1}{m-1}}},
\end{align}}where the variables follow those based in Proposition~\ref{Proposition:OptimalPrefetcher:FF}.}
\end{lemma}
\begin{proof}
See Appendix~\ref{App:Threshold}. 
\end{proof}

\begin{algorithm}[t]
\caption{Finding the approximate prefetching-task set $\widetilde{\mathcal{S}}_n$ in  slot  $n$ for fast fading.}\label{Algorithm:FindingPS:FF}
\begin{algorithmic}[1]
\State Arranging the CTs in descending order in terms of the prefetching-priority function $\delta$ in Definition~\ref{definition:PPF}.
\State Setting $\ell=|\widetilde{\mathcal{S}}_n|$, where $\widetilde{\mathcal{S}}_n=\l\{\ell \in \mathbb{N} \l|s_{n-1}^{*}(\ell)>0, \ell=1, \cdots, L \r.\r\}.$ 
\While{$|\widetilde{\mathcal{S}}_n| < L$}
   \State Estimate the threshold $\eta_{N_P}$ from  \eqref{Eq:AggressivePrefetchingBits:FF} or \eqref{Eq:ConservativePrefetchingBits:FF} corresponding to an aggressive or conservative prefetching policy, respectively. 
   \State Compute $\boldsymbol\alpha^{*}$ using the threshold $\eta_{N_P}$ and  \eqref{Eq:SumofOptimalDecisionVector:FF}.  
   \State Count the number of positive elements in $\boldsymbol\alpha^{*}$, 
   namely $|\boldsymbol\alpha^{*}\succ 0|$. 
   \If{$|\boldsymbol\alpha^{*}\succ 0|=|\widetilde{\mathcal{S}}_n|$}
   \State \textbf{break}
   \EndIf 
   \State $\ell = \ell+1$	and $\widetilde{\mathcal{S}}_n = \widetilde{\mathcal{S}}_n\bigcup \{\ell\}$.
\EndWhile
  \State \textbf{return} $\widetilde{\mathcal{S}}_n$.
\end{algorithmic}
\end{algorithm}

\begin{figure*}[t!]
    \centering
    \subfigure[Expected energy consumption vs. TD size]{
     \includegraphics[width=3.45in]{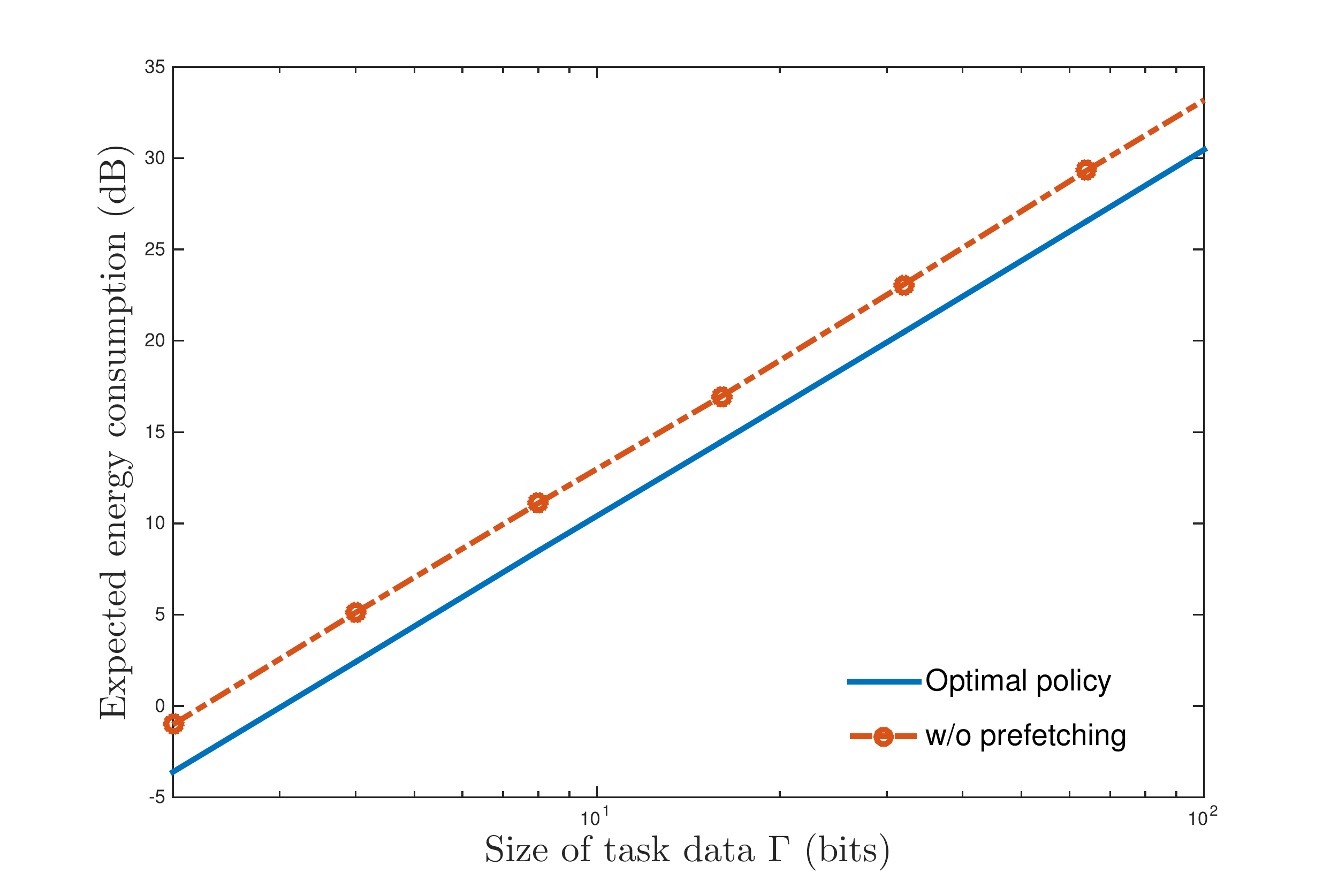}
     \label{Graph2_1}
     } 
    \subfigure[Expected energy consumption vs. number of CTs]{
      \includegraphics[width=3.45in]{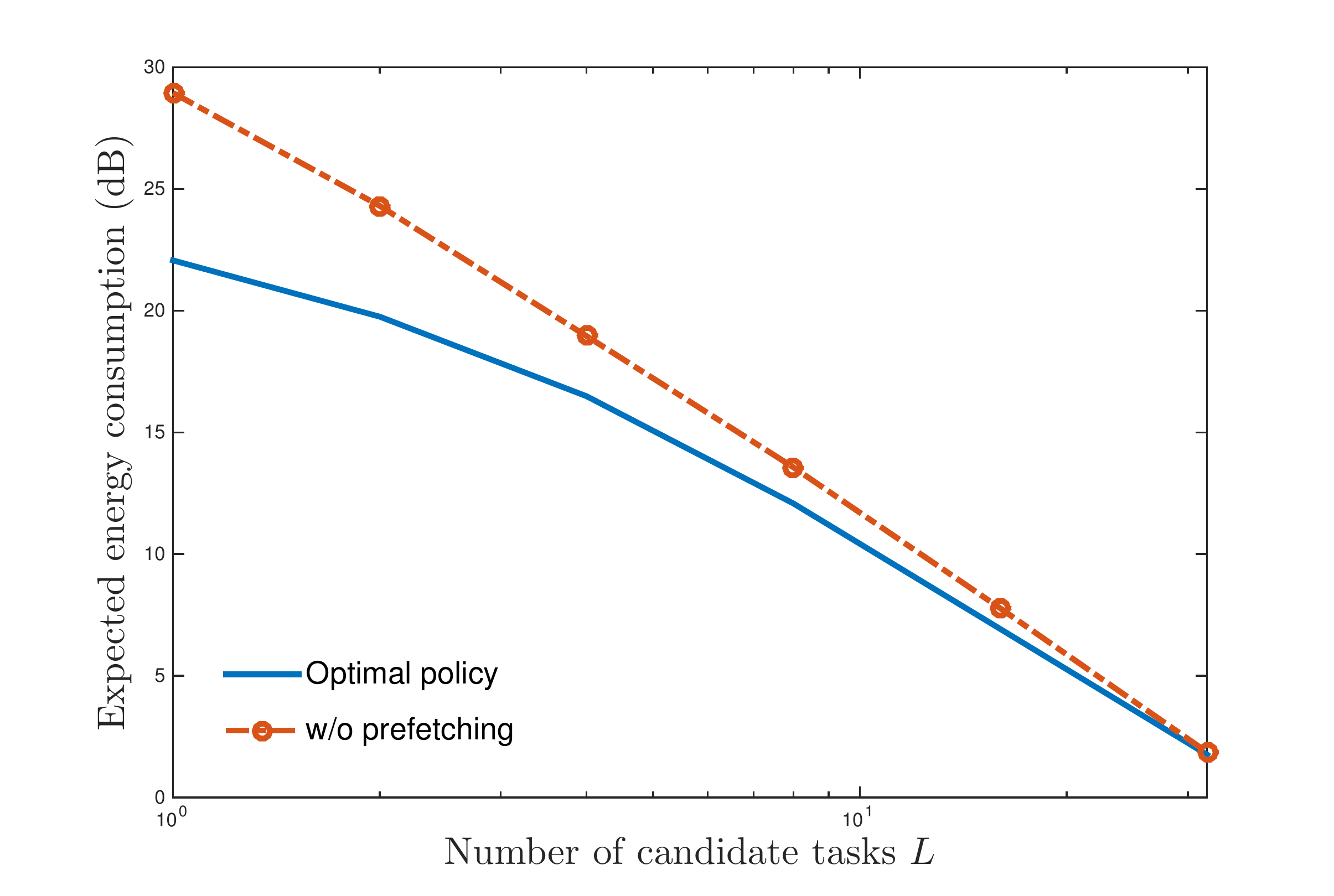}
      \label{Graph2_2}
      }
      \\
      \subfigure[Expected energy consumption vs. latency requirement]
       {
      \includegraphics[width=3.45in]{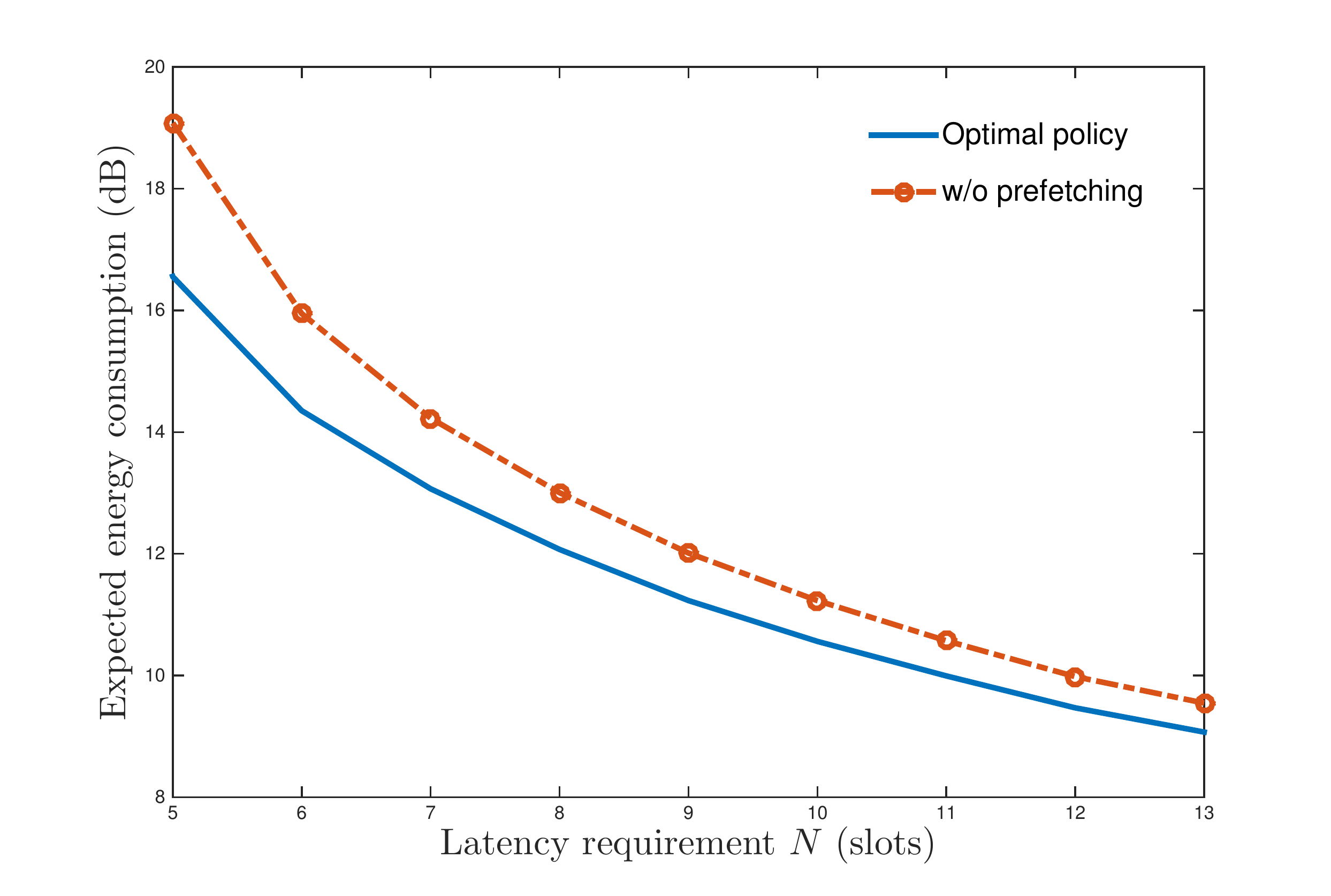}
      \label{Graph2_3}
      }
       \subfigure[Expected energy consumption vs. prefetching duration]
      {
      \includegraphics[width=3.45in]{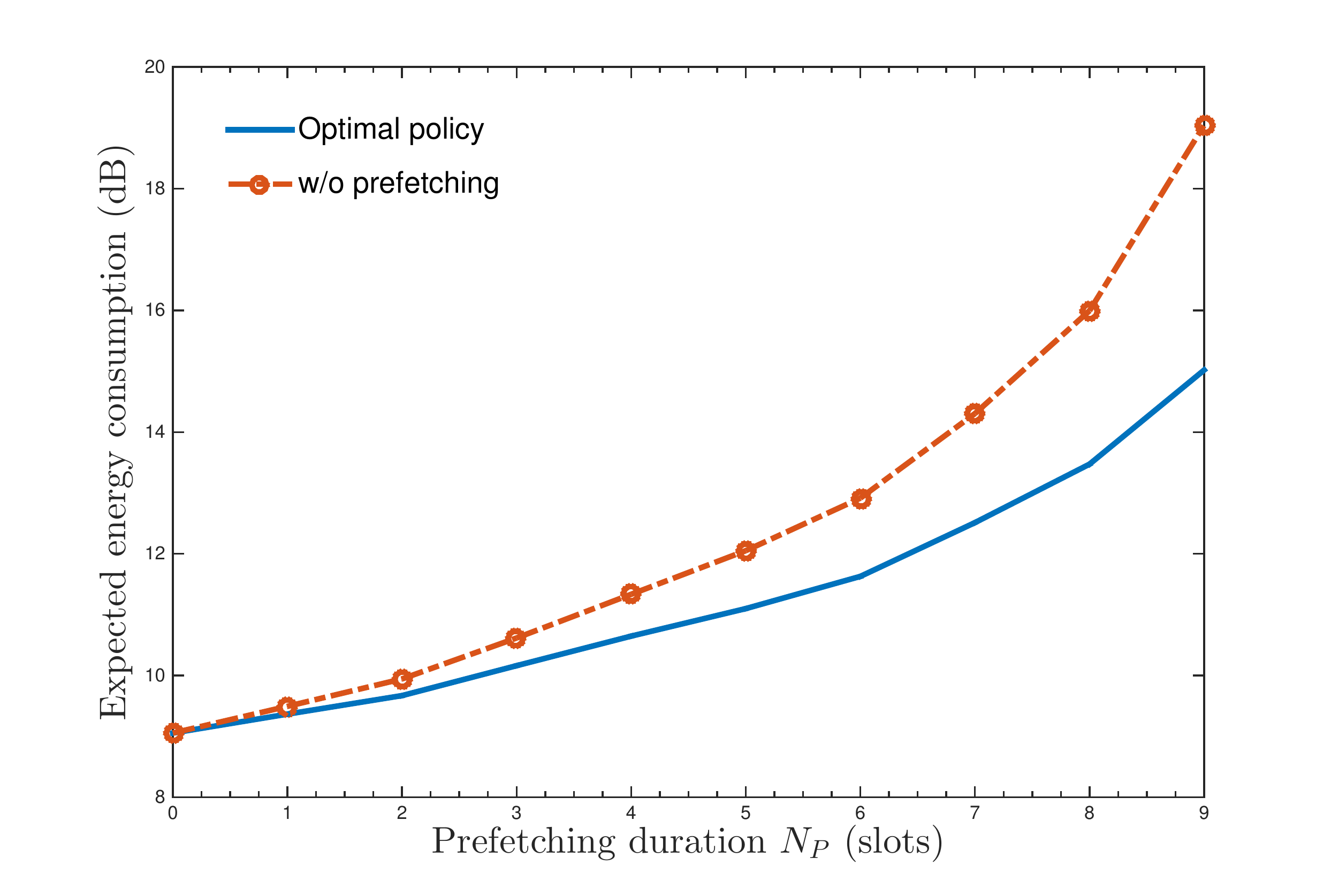}
      \label{Graph2_4}
       } 
      \caption{The effects of MCO parameters on the (mobile) expected energy consumption for the case of slow fading: (a) TD size $\Gamma$ with $N_P=4$, $N=5$ and  $L=4$; (b) the number of CTs  $L$  with $N_P=4$, $N=5$ and $\Gamma=20$; (c)   the latency requirement $N$  (in slot)  with $N_P=4$, $L=4$ and $\Gamma=20$; (d)  the prefetching duration $N_P$ (in  slot) with $N=10$, $L=4$ and $\Gamma=20$. }
    \label{Graph1}
\end{figure*}

The mentioned  difficulty in computing $\eta_{N_P}$ can be overcome by approximation. First, the future channel gains are approximated by their expectation, namely  $g_{k}\approx \mathbb{E}[g]$ for $k = n+1, n+2, \cdots, N_P$. The second approximation is on the following term in \eqref{Eq:FinalThreshold:FF} by its upper bound, which is obtained by applying Jensen's inequality to \eqref{Proposition:OptmalPrefetcher3}:

{\small 
\begin{align}\label{Eq:JensenBound}
&\l(\frac{1}{\zeta_{N-k}(\mathcal{S})}\r)^{\frac{1}{m-1}}\approx\nonumber\\
& \l\{
\begin{aligned}
&\mathbb{E}[g]^{\frac{1}{m-1}}+\l(\frac{1}{\zeta_{N-k-1}(\mathcal{S})}\r)^{\frac{1}{m-1}}, 
&& \textrm{$k<N_P-1$,}\\
&\mathbb{E}[g]^{\frac{1}{m-1}}+\l(\frac{1}{\xi_{N-N_P}}\r)^{\frac{1}{m-1}}\sum_{\ell \in \mathcal{S}} p(\ell)^{-\frac{1}{m-1}},  
&& \textrm{$k=N_P-1$.}
\end{aligned}
\r.
\end{align}}
Substituting the first approximation, namely $g_{k}\approx \mathbb{E}[g]$, gives 
{\small
\begin{align}
&\l(\frac{1}{\zeta_{N-k}(\mathcal{S})}\r)^{\frac{1}{m-1}} \approx \nonumber\\
&\l\{
\begin{aligned}
&(g_{k+1})^{\frac{1}{m-1}}+\l(\frac{1}{\zeta_{N-k-1}(\mathcal{S})}\r)^{\frac{1}{m-1}}, 
&& \textrm{$k<N_P-1$,}\\
&(g_{k+1})^{\frac{1}{m-1}}+\l(\frac{1}{\xi_{N-N_P}}\r)^{\frac{1}{m-1}}\sum_{\ell \in \mathcal{S}} p(\ell)^{-\frac{1}{m-1}},  
&& \textrm{$k=N_P-1$.}
\end{aligned}
\r. 
\end{align}}
\noindent
Substituting this result into \eqref{Eq:FinalThreshold:FF} causes a series of cancelation between product terms, yielding: 
 \begin{align}\label{Eq:AggressivePrefetchingBits:FF}
 \eta_{N_P}\approx &\frac{\sum_{\ell\in\mathcal{S}} \rho_n(\ell){\l(\frac{1}{\xi_{N-N_P}}\r)^{\frac{1}{m-1}}}}{(g_n)^{\frac{1}{m-1}}+\l(\frac{1}{\zeta_{N-n}(\mathcal{S})}\r)^{\frac{1}{m-1}}}.
 \end{align}
 The algorithm in Algorithm~\ref{Algorithm:FindingPS:FF} computes the approximate prefetching-task set $\widetilde{\mathcal{S}}_n$ using \eqref{Eq:AggressivePrefetchingBits:FF} in a similar way as Algorithm~\ref{Algo:SlowFade}. The resultant sub-optimal prefetching policy obtained using Proposition~\ref{Proposition:OptimalPrefetcher:FF} is called an \emph{aggressive} policy since the earlier approximation of fading channel gains by their expectation encourages aggressive prefetching. 

An alternative sub-optimal policy can be designed by approximating the future channel gains by zero: $g_{k}\approx 0$ for $k = n+1, n+2, \cdots, N_P$. Then it follows from \eqref{Eq:FinalThreshold:FF} that 
\begin{align}\label{Eq:ConservativePrefetchingBits:FF}
 \eta_{N_P}\approx \frac{\sum_{\ell\in\mathcal{S}} \rho_n(\ell)\l(\frac{1}{\zeta_{N-n}(\mathcal{S})}\r)^{\frac{1}{m-1}}}{\l((g_n)^{\frac{1}{m-1}}+\l(\frac{1}{\zeta_{N-n}(\mathcal{S})}\r)^{\frac{1}{m-1}}\r)\sum_{\ell \in\mathcal{S}}p(\ell)^{-\frac{1}{m-1}}}.
\end{align}
The alternative policy can be obtained by   replacing the approximation in \eqref{Eq:AggressivePrefetchingBits:FF} with \eqref{Eq:ConservativePrefetchingBits:FF}. The result  is called a  {\it conservative}  policy since the earlier approximation of future channel gains by zero discourages prefetching. 

\begin{figure*}[t!]
    \centering
   \subfigure[Expected energy consumption vs. TD size]{
     \includegraphics[width=3.45in]{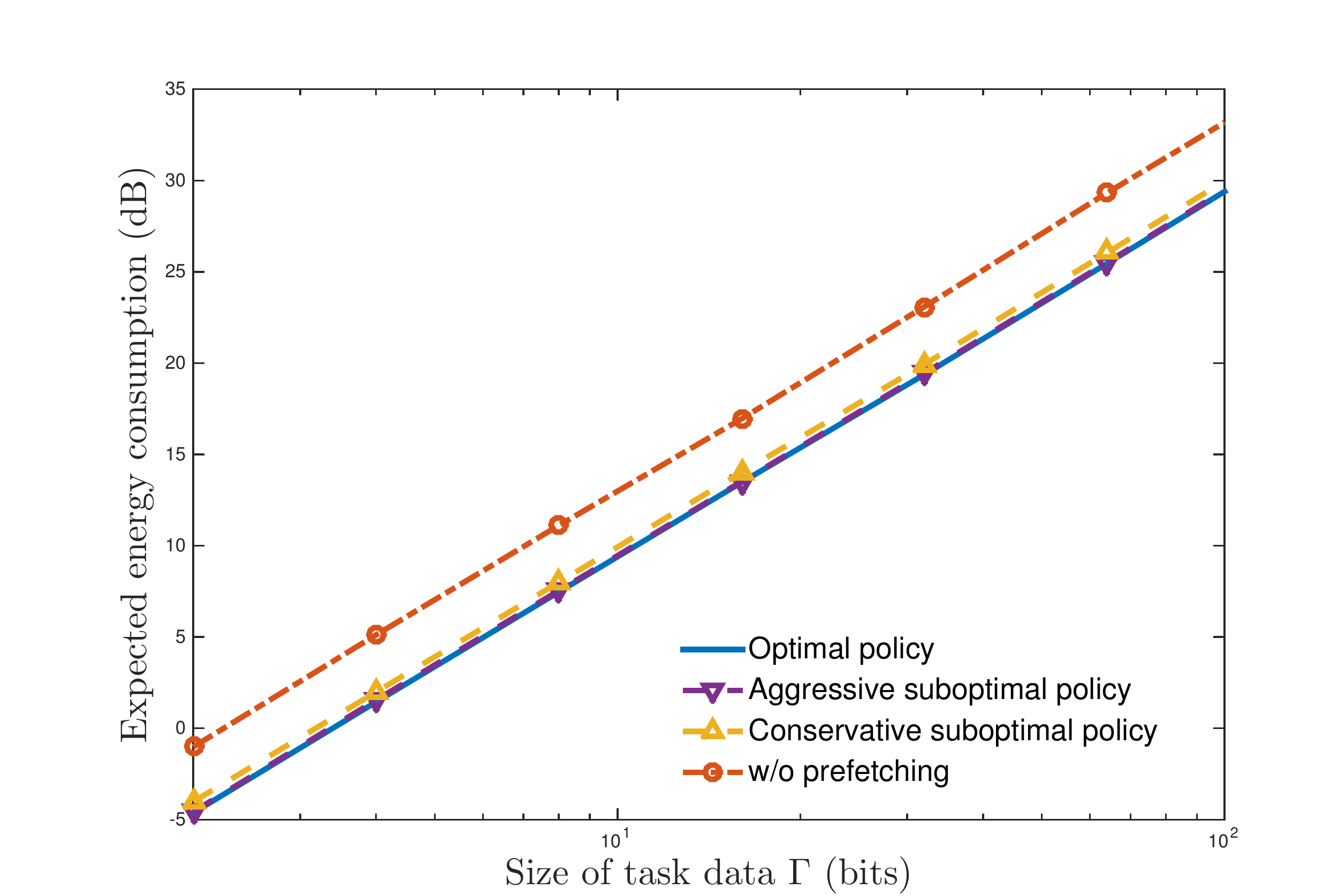}
     \label{Graph2_1}
     } 
    \subfigure[Expected energy consumption vs. number of CTs]{
      \includegraphics[width=3.45in]{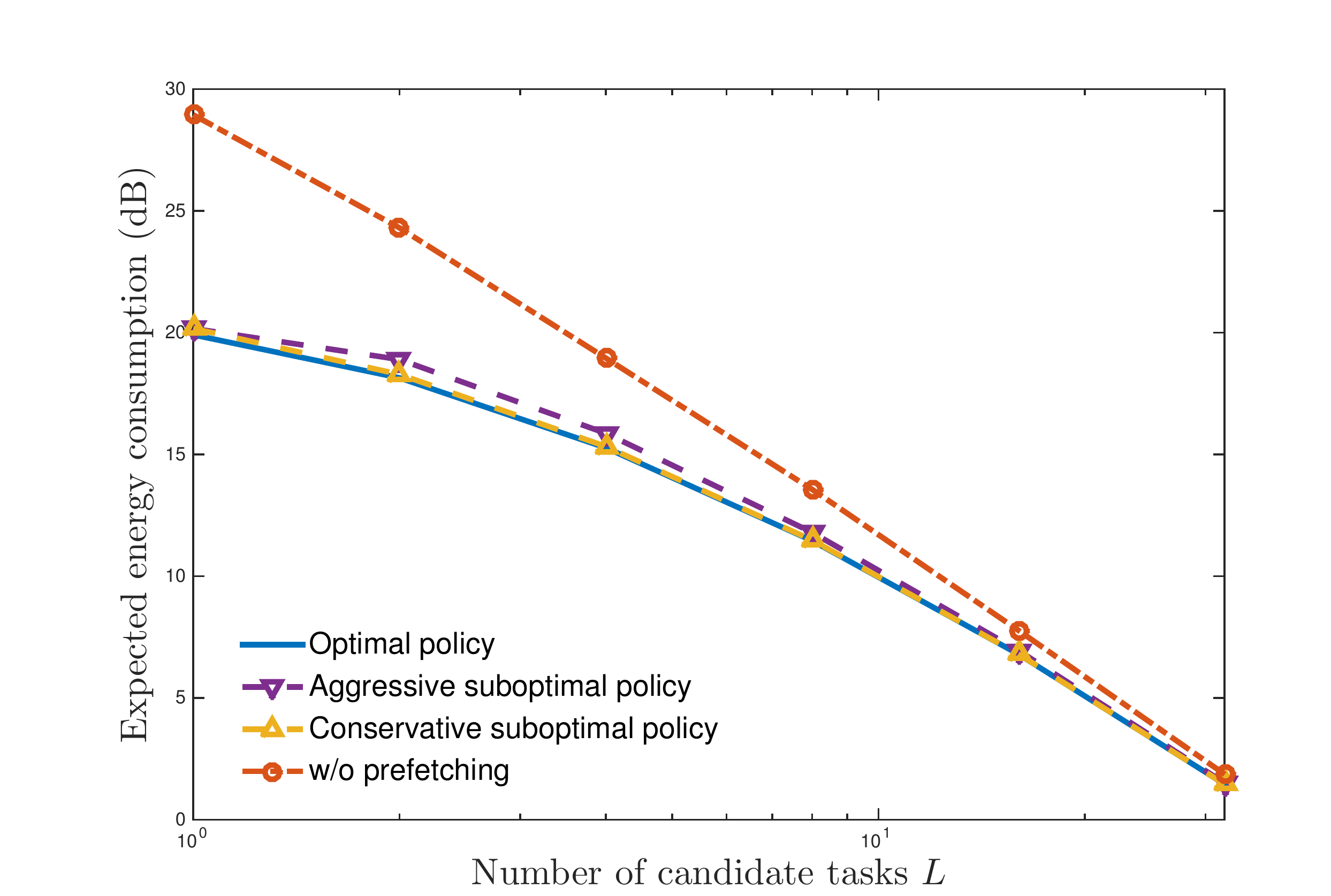}
      \label{Graph2_2}
      }
      \\
      \subfigure[Expected energy consumption vs. latency requirement]
       {
      \includegraphics[width=3.45in]{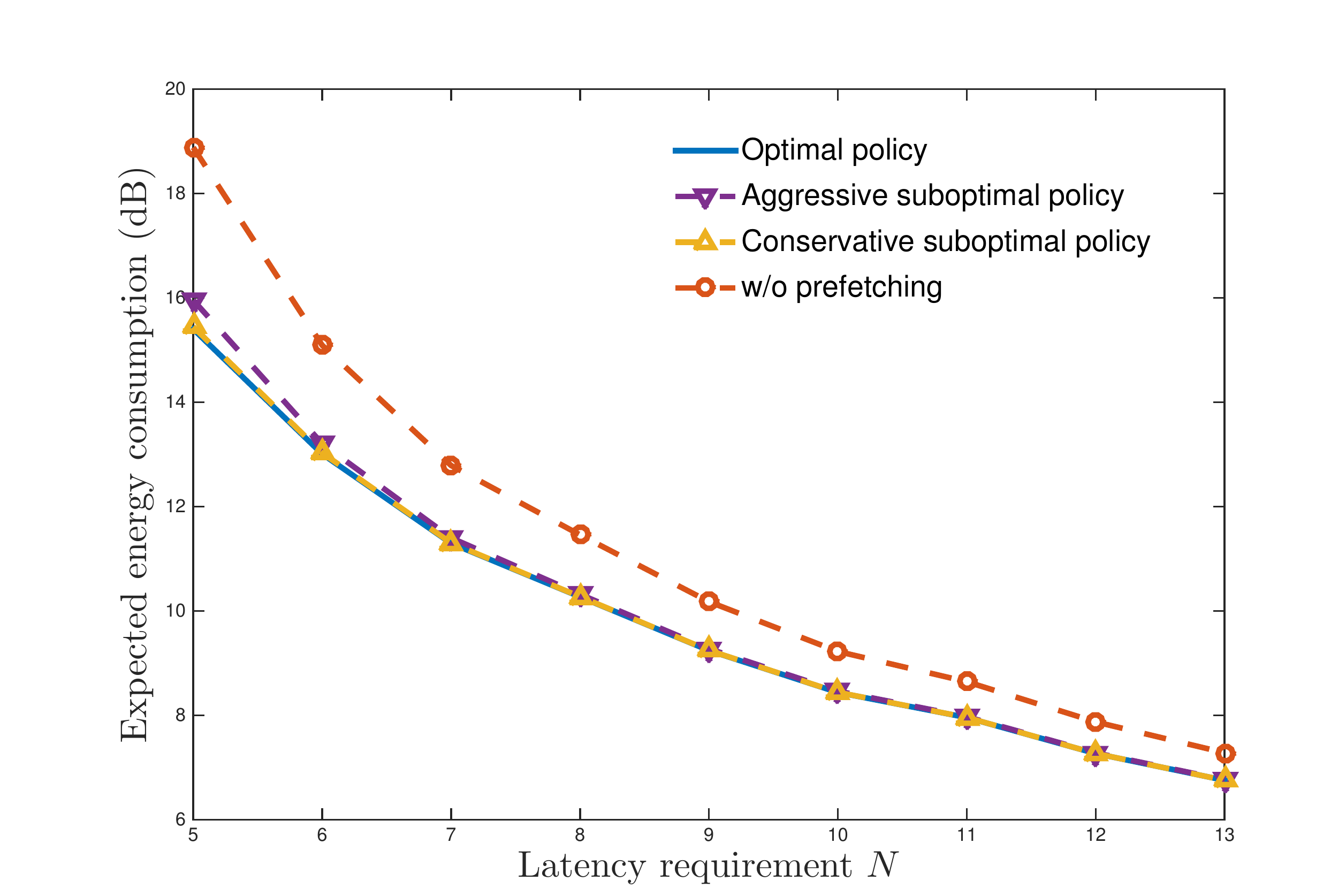}
      \label{Graph2_3}
      }
       \subfigure[Expected energy consumption vs. prefetching duration]
      {
      \includegraphics[width=3.45in]{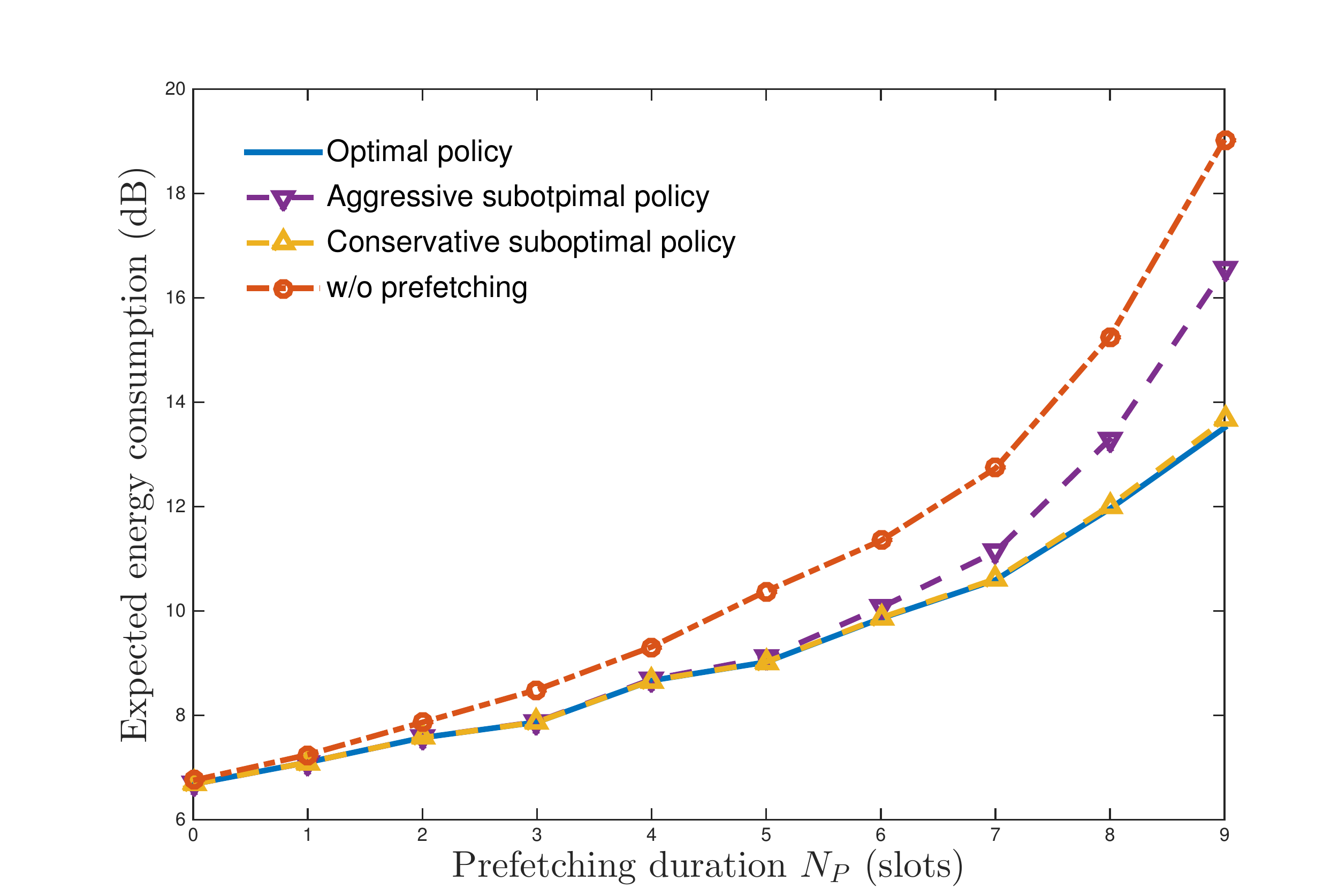}
      \label{Graph2_4}
       } 
      \caption{The effects of MCO parameters on the (mobile) expected energy consumption for the case of fast fading: (a) TD size $\Gamma$ with $N_P=4$, $N=5$ and  $L=4$; (b) the number of CTs  $L$  with $N_P=4$, $N=5$ and $\Gamma=20$; (c)   the latency requirement $N$  (in slot)  with $N_P=4$, $L=4$ and $\Gamma=20$; (d)  the prefetching duration $N_P$ (in  slot) with $N=10$, $L=4$ and $\Gamma=20$. }
    \label{Graph2}
\end{figure*}

\subsection{Prefetching Gain for Fast Fading}

The prefetching gain defined in Definition~\ref{Def:PrefGain} can be quantified  for fast fading as follows. 
\begin{proposition} \label{Theorem:PrefetchingGain:FF} \emph{(Prefetching Gain for Fast Fading).   
Given fast fading, the prefetching gain $G_P$~is 
\begin{align}
G_P > \l[\frac{N-N_P (1-L^{-\frac{m}{m-1}})}{N-N_P }\r]^{m-1},
\end{align}
regardless of $\{\boldsymbol{\gamma}, \boldsymbol{p}\}$.}
\end{proposition}
\begin{IEEEproof}
See Appendix \ref{App:PrefetchingGain:FF}. 
\end{IEEEproof}

\begin{remark} \emph{ (Effects of Fast Fading on the  Prefetching Gain)
The prefetching gain $G_P$ for fast fading is always larger than the slow-fading counterpart.   
This suggests that fasting fading enables opportunistic prefetching that exploits channel temporal diversity for enhancing the gain. In case of slow fading, on the other hand, it is hard to achieve the diversity gain because the equal bit allocation is proved to be optimal in terms of energy efficiency.} 
\end{remark}

\section{Simulation Results}\label{Section:Simulation}

Simulation results are presented for  evaluating the performance of MCO with live prefetching. In the simulation, 
the channel gain $g_n$  follows the Gamma  distribution with the integer shape parameter $k > 1$ and the probability density function $f_g(x)=\frac{x^{k-1}e^{-k x}}{(1/k)^k \Gamma(k)}$ where the Gamma function $\Gamma(k)=\int_{0}^{\infty}x^{k-1}e^{-x} dx$ and  the  mean $\mathbb{E}\l[{g_n}\r]=1$. 
The parameter is set as $k = 2$ in simulation unless stated otherwise.  
The monomial order of  the energy consumption model in \eqref{EnergyModel} is set as $m=2$. The transition probabilities $\{p(\ell)\}$ and the data size $\{\gamma(\ell)\}$ of CTs are generated based the uniform distribution. Furthermore, for the case of fast fading, the optimal policy without any closed form  is computed using non-causal CSI and used as a benchmark for evaluating the performance of sub-optimal policies.

The curves of the expected energy consumption (in dB)  are plotted against  different MCO parameters in separate subfigures of  Fig.~\ref{Graph1}. The parameters include 
the task-data size $\Gamma$,  number of CTs $L$, the latency requirement $N$ in slot, and prefetching duration $N_P$ in slot. 
  Since the vertical axes of the figures are on the   logarithmic scale, the vertical gaps between the curves measure the expected prefetching gain $\mathbb{E}[G_P]$ over the random task-transmission probabilities. First, one can observe that the prefetching gain is always larger than one in all sub-figures. Second, with both axes on  the logarithmic scale,    the straight lines in Fig.~\ref{Graph1}(a) imply  that the expected energy consumption increases as a monomial of the TD size, which agrees with energy consumption model in \eqref{EnergyModel}. Moreover, the prefetching gain is approximately $2.3$ dB. Third, the expected prefetching gain is observed from Fig.~\ref{Graph1}(b) to be more than $5$ dB when the number of CTs $L$ is small but the gain diminishes as $L$ increases. 
The reason is that a larger set of CTs with uniform likelihoods makes it more difficult to accurately predict the next task as well as increase the amount of redundant prefetched data. Fourth, Fig.~\ref{Graph1}(c) shows that the prefetching gain diminishes as the latency requirement is relaxed (corresponding to increasing $N$). In other words, prefetching is beneficial in scenarios with relatively stringent latency requirements. Last, the gain grows with the prefetching duration which is aligned with intuition.

Next, a similar set of curves as those in  Fig.~\ref{Graph1} are plotted in Fig.~\ref{Graph2} with identical parametric values but for the case of fast fading.  The curves corresponding to the sub-optimal prefetching polices  designed in Section~\ref{Section:FastFadingPart} are also plotted in Fig.~\ref{Graph2}.  Similar observations as those for Fig.~\ref{Graph1} can be also made for Fig.~\ref{Graph2}. The new observation is that both the proposed sub-optimal polices achieve close-to-optimal performance for the considered ranges of TD size, number of CTs and latency requirement [see Fig.~\ref{Graph1}(a)-(c)]. However,  the conservative prefetching is preferred to the aggressive one for the case where the prefetching duration is large as one can see in Fig.~\ref{Graph1}(d). 

Last, the effects of fading on the expected prefetching gain are shown in Fig.~\ref{Fig:PrefetchingGain} where the shape parameter $k$ of the fading distribution controls the level of channel randomness. It can be observed that the gain reduces as $k$ increases, corresponding to a decreasing level of channel randomness. For large $k$, the gain saturates at the constant corresponding to slow fading. 

\begin{figure}[t]
\centering
\includegraphics[width=3.6in]{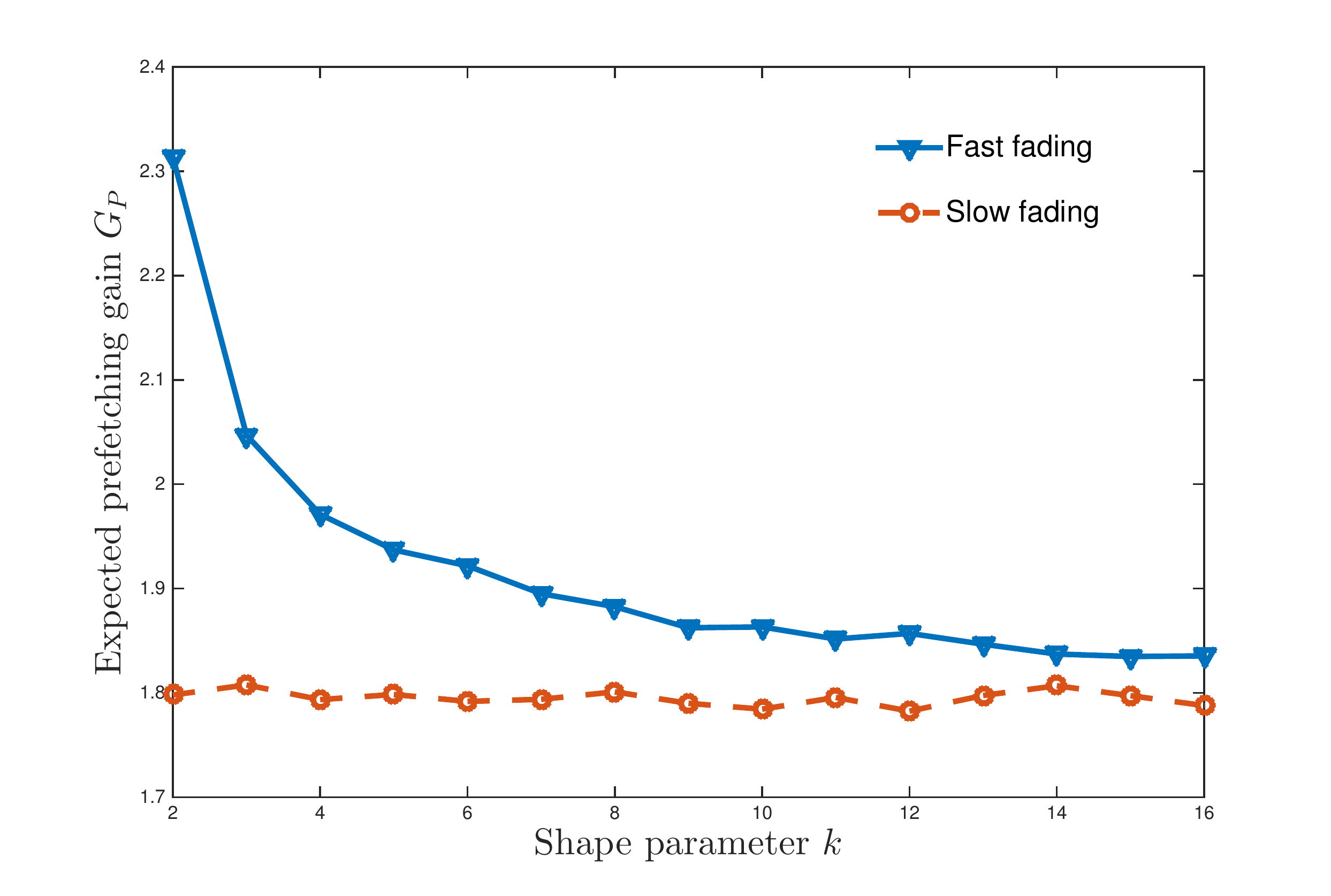}
\caption{The effects of channel randomness on the expected prefetching gain $\mathbb{E}[G_P]$ where randomness reduces with the increasing  shape parameter $k$ of the fading distribution. The parameters are set as 
$\Gamma=20$, $N=5$, $N_P=4$ and $L=4$.} \vspace{-20pt} \label{Fig:PrefetchingGain}
\end{figure}

\section{Concluding Remarks}\label{Section:Conclusion}

A novel architecture of live prefetching for mobile computation offloading has been proposed in this paper 
to enable  prefetching based on  task-level computation   prediction and its simultaneous operation with cloud computing. Given stochastic sequential tasks, the optimal and sub-optimal prefetching polices have been designed to minimize mobile energy consumption. Specifically, the designed policies are capable of  selecting  tasks for prefetching and controlling prefetched data sizes for both the cases of slow and fast fading. The simple threshold based structure of the policies is derived, which enables low-complexity realtime operation. Comprehensive simulation shows that  live prefetching achieves significant mobile energy reduction compared with conventional schemes without prefetching. 

It is worth noting that live prefetching can be applied in 5G network architecture 
by facilitating real-time signal processing and radio resource management, 
requiring delay sensitive information from mobiles as an input data, e.g. wireless channel and mobility states.    
Applying live prefetching enables mobiles to prefetch these data within given energy and time budgets
by leveraging the sophisticated prediction of the subsequent status.
In so doing, the network is able to perform collaborative decisions 
such as interference coordination and mobility management.  

This work can be extended in several interesting research directions. First, live prefetching can be integrated with advanced computing techniques such as the parallel computing and  pipelining to achieve higher mobile energy efficiencies but the designs are more complex. Second, live prefetching can be also jointly designed with more sophisticated wireless transmission techniques such as multi-antenna transmission and multi-cell cooperation. 
 Last, it is interesting to characterize the performance gain of large-scale cloud-computing networks due to live prefetching.

\appendix 
\subsection{Proof of Proposition \ref{Proposition:OptimalPrefetchingPolicy:SF}}
\label{App:OptimalPrefetchingPolicy:SF}
Define the Lagragian function for Problem \ref{TwoStageOptimization:SlowFading} as
\begin{align}
L&=\frac{\l(\sum_{\ell=1}^{L} \alpha(\ell)\r)^m}{(N_P)^{m-1}}+\frac{\sum_{\ell=1}^L p(\ell)(\gamma(\ell)-\alpha(\ell))^m}{(N-N_P)^{m-1}}\nonumber\\
&+\sum_{\ell=1}^L \mu(\ell) (\alpha(\ell)-\gamma(\ell)), \nonumber
\end{align}
where $\mu(\ell)\geq0$ are Lagragian multipliers. Since Problem \ref{TwoStageOptimization:SlowFading} is a convex optimization problem, the following KKT conditions are necessary and sufficient for optimality:
\begin{alignat}{2} 
& 
m \frac{\l(\alpha^*_{\Sigma}\r)^{m-1}}{(N_P)^{m-1}} - 
m p(\ell)\frac{\l(\gamma(\ell)-\alpha(\ell)\r)^{m-1}}{(N-N_P)^{m-1}}+\mu(\ell)\geq 0,  
\label{Eq:Problem3KKT1} 
\\
& \alpha(\ell)\l(m \frac{\l(\alpha^*_{\Sigma}\r)^{m-1}}{(N_P)^{m-1}} - 
m p(\ell)\frac{\l(\gamma(\ell)-\alpha(\ell)\r)^{m-1}}{(N-N_P)^{m-1}}+\mu(\ell)\r)=0,  
\label{Eq:Problem3KKT2} \\
& \mu(\ell)(\alpha(\ell)-\gamma(\ell))=0,  
 \label{Eq:Problem3KKT3} 
\end{alignat}
If $\mu(\ell)$ is postive,  $\alpha^*(\ell)$ is equal to $\gamma(i)$ due to the slackness condition of \eqref{Eq:Problem3KKT3}, yielding \eqref{Eq:Problem3KKT2}  strictly positive. It violates condition \eqref{Eq:Problem3KKT2} and 
the optimal multiplier $\mu(\ell)$ is thus zero for all $\ell$.

If $\alpha^*(\ell)$ is zero, \eqref{Eq:Problem3KKT1} leads to 
$\gamma(\ell)< p(\ell)^{-\frac{1}{m-1}}\frac{N-N_P}{N_P} \alpha^*_{\Sigma}$.
Otherwise, 
$\alpha^*(\ell)$ is equal to $\gamma(\ell)-p(\ell)^{-\frac{1}{m-1}}\frac{N-N_P}{N_P} \alpha^*_{\Sigma}$.
Combining these two completes the proof.
 
\subsection{Proof of Corollary \ref{Corollary:UniqueOptimalPrefetchingVector:SF}}\label{App:UniqueOptimalPrefetchingVector:SF}
First, it is straightforward to show $\boldsymbol\alpha^*=\boldsymbol\gamma$ when $N=N_P$. 
Let us consider $N>N_P$. 
Two equations of \eqref{Eq:Optimalalpha:SF} and
\eqref{Eq:TotalPrefetchedBits:SF} construct a nonlinear system in the two unknown $\boldsymbol\alpha^*$ and $\alpha^*_{\Sigma}$.
Eq.~\eqref{Eq:Optimalalpha:SF} can be inverted~to  
\begin{align}
\alpha^*_{\Sigma}(\boldsymbol\alpha^*)=(\gamma(\ell)-\alpha^*(\ell))p(\ell)^{\frac{1}{m-1}}\frac{N_P}{N-N_P},\quad \ell \in \mathcal{S}.
\end{align}
It is easily proved that $\alpha^*_{\Sigma}(\boldsymbol\alpha^*)$ 
is a continuous and monotone decreasing function of $\boldsymbol\alpha^*\in (0, \boldsymbol\gamma)$ from 
$\alpha^*_{\Sigma}(0)=\max_{\ell}\l\{\gamma(\ell)p(\ell)^{\frac{1}{m-1}}\r\}$$\frac{N_P}{N-N_P}$ to $\alpha^*_{\Sigma}(\boldsymbol\gamma)=0$. In addition, $\alpha^*_{\Sigma}(\boldsymbol\alpha)=\sum_{\ell=1}^L \alpha^*(\ell)$ 
is a continuous and monotone increasing function of 
$\boldsymbol\alpha^*\in (0, \boldsymbol\gamma)$ from $\alpha^*_{\Sigma}(0)=0$ to $\alpha^*_{\Sigma}(\boldsymbol\gamma)=\sum_{\ell=1}^L\gamma(\ell)$. Thus, a unique optimal prefetching vector $\boldsymbol\alpha^*$ exists in the range of
$0\leq \boldsymbol\alpha^*\leq \boldsymbol\gamma$.

\subsection{Proof of Proposition \ref{Theorem:PrefetchingGain:SF}}\label{App:PrefetchingGain:SF}
The energy consumption of MCO without prefetching depends on fetching the TD of selected CT~$\ell$ during $(N-N_P)$  slots,  
\begin{align}
\sum_{\ell = 1}^L {E}^*_{D}(\gamma(\ell))p(\ell)&=\lambda\frac{\sum_{\ell=1}^L p(\ell) \gamma(\ell)^m}{(N-N_P)^m}\nonumber\\
&=\lambda \frac{\sum_{\ell=1}^L  \gamma(\ell)^m}{L{(N-N_P)^m}}\geq\frac{\lambda \l(\frac{\Gamma}{L}\r)^m }{(N-N_P)^m}
\end{align}
where the equality holds when $\gamma(\ell)=\frac{\Gamma}{L}$ for all $\ell$. 
On the other hands, the energy consumption of MCO with live prefetching~is 
{\small
\begin{align}
&E_P(\boldsymbol\alpha^*)+\sum_{\ell\in \mathcal{S}} p(\ell) E^*_D(\beta(\ell))\nonumber\\
=& \frac{ \lambda \alpha^*_{\Sigma}}{(N_P)^m}+\frac{\lambda\l(\sum_{\ell \in \mathcal{S}}p(\ell)^{-\frac{1}{m-1}}\r) 
(\frac{N-N_P}{N_P}\alpha^*_{\Sigma})^m}{(N-N_P)^{m-1}}+\frac{\lambda\sum_{\ell \notin \mathcal{S}}\gamma(\ell)^m}{(N-N_P)^{m-1}},\nonumber
\end{align}
}where $\alpha^*_{\Sigma}=\frac{\sum_{\ell \in \mathcal{S}} \gamma(\ell)}{1+\frac{N-N_P}{N_P}\l(\sum_{\ell \in \mathcal{S}}p(\ell)^{-\frac{1}{m-1}}\r)}$. Noting that $\frac{N-N_P}{N_P}\alpha^*_{\Sigma} \geq \gamma(\ell)$ for $\ell \notin \mathcal{S}$,  the above  is upper bounded~as
\begin{align}
&E_P(\boldsymbol\alpha^*)+\sum_{\ell\in \mathcal{S}} p(\ell) E^*_D(\beta(\ell))\nonumber\\
\leq & \lambda \frac{\alpha^*_{\Sigma}}{(N_P)^m}+\lambda\frac{\l(\sum_{\ell=1}^L p(\ell)^{-\frac{1}{m-1}}\r) 
(\frac{N-N_P}{N_P}\alpha^*_{\Sigma})^m}{(N-N_P)^{m-1}},
\end{align}
where the equality is satisfied when the prefetching task set $\mathcal{S}$ contains all CTs. 
In addition, it is maximized at the maximum value of $\alpha^*_{\Sigma}$ that can be achieved when $p(\ell)=\frac{1}{L}$ and $\gamma(\ell)=\frac{\Gamma}{L}$ for every~$\ell$. 
The denominator and the numerator of $G_P$ in \eqref{Eq:PrefetchingGainDefinition}
are respectively minimized and maximized at the above setting, completing the proof.
\subsection{Proof of Proposition~\ref{Proposition:OptimalDemandFetcher:FF}}\label{App:OptimalDemandFetcher:FF}
The expected energy consumption at slot $N$, 
$\bar{J}_{N}(l_{N}(\ell), g_{N})$ is $\lambda \mathbb{E}\l[\frac{1}{g_{N}}\r]{\rho_{N}(\ell)}^{m}=\lambda \xi_{1}{\rho_{N}(\ell)}^{m}$. 
According to Problem \ref{OptimalBit}, the optimal number of transmitted bits at slot $N-1$, $b_{N-1}^*$ is 
\begin{align}\label{Proof:Proposition1_1}
b_{N-1}^*&=\arg \min_{b_{N-1}} \l(\lambda \frac{\l(b_{N-1}\r)^{m}}{g_{N-1}}
+\lambda \xi_{1}({\rho_{N-1}(\ell)-b_{N-1}})^{m}\r)\nonumber\\
&=\frac{\rho_{N-1}(\ell) (g_{N-1})^{\frac{1}{m-1}}}{(g_{N-1})^{\frac{1}{m-1}}+ \l(\frac{1}{\xi_{1}}\r)^{\frac{1}{m-1}}},
\end{align}
and the cost-to-go function $\bar{J}_{N-1}\l(\rho_{N-1}(\ell),g_{N-1}\r)=\mathbb{E}_{g_{N-1}}\l[J_{N-1}(\rho_{N-1}(\ell),g_{N-1})\r]$ is
\begin{align}\label{Proof:Proposition1_2}
\bar{J}_{N-1}(l_{N-1}(\ell))
=&\lambda \rho_{N-1}(\ell)^{m}\mathbb{E}\l[\tfrac{1}{\l\{{(g_{N-1})^{\frac{1}{m-1}}+\l(\frac{1}{\xi_{1}}\r)^{\frac{1}{m-1}}}\r\}^{m-1}}\r]\nonumber\\
=&\lambda \rho_{N-1}(\ell)^{m}\xi_{2}.
\end{align}
Calculating the above  repeatedly, we can derive 
the optimal transmitted bits at slot $n$,  $b_n^{*}$,
 \begin{align}\label{Proof:Proposition1_3}
\!\!\!\!b_{n}^{^*}(l_n(\ell), g_n)&=\arg \!\!\min_{0 \leq b_{n} \leq \rho_{n}(\ell)} \!\! \l[\lambda \frac{\l(b_{n}\r)^{m}}{g_{n}}+\lambda ({\rho_{n}(\ell)-b_{n}})^{m}\xi_{n+1}\r]\nonumber\\
&=\frac{\rho_{n}(\ell) g_{n}^{\frac{1}{m-1}}}{(g_{n})^{\frac{1}{m-1}}+ \l(\frac{1}{\xi_{N-n}}\r)^{\frac{1}{m-1}}},
\end{align}
and $\bar{J}_{n}(\rho_{n}(\ell))=\lambda \l(\rho_{n}(\ell)\r)^{m}\xi_{N-n+1}$.
 After  inserting $\beta(\ell)$ and $N_P+1$ into $\rho_{n}(\ell)$ and $n$ respectively, the proof is completed.

\subsection{Proof of Corollary~\ref{Corollary:DemandFetchingEnergyBound:FF}}\label{App:DemandFetchingEnergyBound:FF}
First, we show the lower bound of $\xi_n$. Noting that $X^{-(m-1)}$ is a convex function of $X$, the following inequality is satisfied:
\begin{align}
\!\!\xi_{n}&=\mathbb{E}\l[\l(g^{\frac{1}{m-1}}+{\l(\frac{1}{\xi_{n-1}}\r)}^{\frac{1}{m-1}}
\r)^{-(m-1)}\r] \nonumber\\
&\overset{(a)}{\geq}\l(\mathbb{E}\l[g^{\frac{1}{m-1}}+{\l(\frac{1}{\xi_{n-1}}\r)}^{\frac{1}{m-1}}\r]\r)^{-(m-1)},\label{Proof:convergence1_1}
\end{align}
where $(a)$ follows from Jensen's inequality. Rearranging \eqref{Proof:convergence1_1} gives
\begin{align}
\!\!\l(\frac{1}{\xi_{n}}\r)^{\frac{1}{m-1}}&
\!\!\leq\mathbb{E}\!\!\l[{g}^{\frac{1}{m-1}}+{\l(\frac{1}{\xi_{n-1}}\r)}^{\frac{1}{m-1}}\r]\nonumber\\
&\!\!=\!\!\mathbb{E}\!\!\l[{g}^{\frac{1}{m-1}}\r]\!\!+\!\!{\l(\frac{1}{\xi_{n-1}}\r)}^{\frac{1}{m-1}}\!\!\leq 
n \mathbb{E}\l[{g}^{\frac{1}{m-1}}\r]\leq n{\mathbb{E}[g]}^{\frac{1}{m-1}},\nonumber
\end{align}
and $\xi_n$ of \eqref{Proposition:xi} is thus lower bounded as 
${\xi_n}\geq n^{-(m-1)}\frac{1}{\mathbb{E}[g]}$
where the equality hold when $n\rightarrow \infty$.

Second, we show the upper bound of $\xi_n$ by using $\bar{\xi}_n$ given~as
{\small
\begin{align}\label{upper_xi}
\bar{\xi}_{n}=
\l\{
\begin{aligned}
& \mathbb{E}\l[\frac{1}{g}\r],
&& \textrm{$n=1$,}\\
&\!\!\l[\l(\frac{1}{\mathbb{E}\l[\frac{1}{g}\r]}\r)^{\frac{1}{m-1}}\!\!\!\!\!\!\!\!+\l(\frac{1}{\xi_{n-1}}\r)^{\frac{1}{m-1}}\r]^{{-(m-1)}},
&& \textrm{$n>1$.}
\end{aligned}
\r.
\end{align}}
It is straightforward to show that 
$\xi_n$ of \eqref{Proposition:xi} is lower bounded by $\bar{\xi}_n$ of \eqref{upper_xi}, 
and then expressed by the following inequality:
\begin{align}
\!\!\l(\frac{1}{\xi_{n}}\r)^{\frac{1}{m-1}}
\!\!\!\!\geq 
\l(\frac{1}{\bar{\xi}_{n}}\r)^{\frac{1}{m-1}}
\!\!\!\!\!\!\!\!=
\l(\frac{1}{\mathbb{E}\l[\frac{1}{g}\r]}\r)^{\frac{1}{m-1}}\!\!\!\!\!\!\!\! +{\l(\frac{1}{\xi_{n-1}}\r)}^{\frac{1}{m-1}}
\!\!\!\!\!\!\!\!\geq 
n \l(\frac{1}{\mathbb{E}\l[\frac{1}{g}\r]}\r)^{\frac{1}{m-1}}, \nonumber
\end{align}
and $\xi_n$ is thus lower bounded as
${\xi_n}\leq n^{-(m-1)}\mathbb{E}\l[\frac{1}{g}\r]$
where the equality hold when $n=1$.
We have the result after substituting  the upper and lower bounds into~\eqref{Proposition:MinimumEnergyForDF}. 
\subsection{Proof of Proposition~\ref{Proposition:OptimalPrefetcher:FF}}\label{App:OptimalPrefetcher:FF}
 
The optimal decision vector $\boldsymbol{s}^*_{N_P}$ and $V_{N_P}(\boldsymbol{\rho}_{N_P}, g_{N_P})$ at slot $N_P$ are respectively
{\small
\begin{align}
\boldsymbol{s}_{N_P}^{*}\!\!&=\!\!\l[\boldsymbol{\rho}_{N_P}-\frac{\sum_{\ell \in \mathcal{S}}\rho_{N_P}(\ell)\l(\frac{1}{\xi_{N-N_P}}\r)^{\frac{1}{m-1}}\boldsymbol{p}^{-\frac{1}{m-1}}}{(g_{N_P})^{\frac{1}{m-1}}\!\!+\!\!\l(\frac{1}{\xi_{N-N_P}}\r)^{\frac{1}{m-1}}\sum_{\ell \in \mathcal{S}}p(\ell)^{-\frac{1}{m-1}}}\r]^+\nonumber,
\end{align}}and
{\small
\begin{align}
&V_{N_P}(\boldsymbol{\rho}_{N_P}, g_{N_P})\nonumber\\=
&\frac{\lambda \l[\sum_{\ell\in \mathcal{S}} \rho_{N_P}(\ell) \r]^{m}}
{\l[(g_{N_P})^{\frac{1}{m-1}}\!\!+\!\!\l(\frac{1}{\xi_{N-N_P}}\r)^{\frac{1}{m-1}}\!\!\l(\sum_{\ell \in \mathcal{S}} p(\ell)^{-\frac{1}{m-1}}\r)\!\r]^{m-1} }\nonumber\\
&+\lambda \xi_{N-N_P}\!\! \l(\sum_{\ell \notin \mathcal{S}}p(\ell) {\rho_{N_P}(\ell)^{m}}\r),\nonumber
\end{align}}where the derivation is in a similar way 
in the proof of Proposition~\ref{Proposition:OptimalPrefetchingPolicy:SF}.
 
Next, consider slot $N_P-1$. The cost-to-go function $\bar{V}_{N_P}(\boldsymbol{\rho}_{N_P})=\mathbb{E}\l[V_{N_P}(\boldsymbol{\rho}_{N_P}, g_{N_P})\r]$ is
\begin{align}\label{Proof:Proposition2Proof7}
\bar{V}_{N_P}(\boldsymbol{\rho}_{N_P})=&{\lambda \l(\sum_{\ell\in \mathcal{S}} \rho_{N_P}(\ell) \r)^{m}}
{\zeta_{N-N_P+1}(\mathcal{S})}\nonumber\\
&+\lambda \xi_{N-N_P} \l(\sum_{i \notin \mathcal{S}}p(\ell) {\rho_{N_P}(\ell)^{m}}\r).
\end{align}
Substituting \eqref{Proof:Proposition2Proof7} into Problem~\ref{OptimizationPrefetcher} of slot $N_P-1$ gives
\begin{align}
\underset{\boldsymbol{s}_{N_P-1}}{\min} \!\!
&\l[ \frac{\l(\sum_{\ell \in \mathcal{S}}s_{N_P-1}(\ell)\r)^{m}}{g_{N_P-1}}
+{ \l(\sum_{\ell\in \mathcal{S}} \rho_{N_P-1}(\ell)- s_{N_P-1}(\ell) \r)^{m}}
\r.\nonumber\\
&\l.\cdot {\zeta_{N-N_P+1}(\mathcal{S})}+ \xi_{N-N_P} \l(\sum_{\ell \notin \mathcal{S}}p(\ell) {\rho_{N_P-1}(\ell)^{m}}\r)\r].\nonumber
\end{align}
Given $\mathcal{S}$, the last term is independent of the optimal decision vector $\boldsymbol{s}_{N_P-1}^*$, 
and  it is then equivalent to find the optimal transmitted bits $b_{N_P-1}^*=\sum_{\ell\in \mathcal{S}} s_{N_P-1}(\ell)$ as
\begin{align}\label{Proof:Proposition2Proof9}
b_{N_P-1}^*&=
\frac{(g_{N_P-1})^{\frac{1}{m-1}}\sum_{\ell \in \mathcal{S}}\rho_{N_P-1}(\ell)}{(g_{N_P-1})^{\frac{1}{m-1}}+\l(\frac{1}{\zeta_{N-N_P+1}(\mathcal{S})}\r)^{\frac{1}{m-1}}    }.
\end{align}
Calculating the above procedure repeatedly, we can make the following optimization problem:   
\begin{align}\label{Proof:Proposition2Proof12}
b_n^*&=\underset{b_{n}}{\arg\min}  
\l[\lambda \frac{(b_{n})^{m}}{g_{n}}
+{\lambda \l(\sum_{\ell\in\mathcal{S}}\rho_{n}(\ell)- b_{n} \r)^{m}}
{\zeta_{N-n}(\mathcal{S})}\r]\nonumber\\
&=\frac{(g_{n})^{\frac{1}{m-1}}\sum_{\ell \in \mathcal{S}}\rho_{n}(\ell)}{(g_{n})^{\frac{1}{m-1}}+\l(\frac{1}{\zeta_{N-n}(\mathcal{S})}\r)^{\frac{1}{m-1}}    }.
\end{align}
The optimal bits $b_n^*$ can be decomposed into multiple $s_n^*(\ell)$ for $\ell\in\mathcal{S}$ 
by balancing the remaining bits~$\rho_{n+1}(\ell)$ as \eqref{Eq:Optimalalpha:SF}.
\subsection{Proof of Corollary~\ref{Corollary:PropertyThreshold}}\label{App:PropertyThreshold}
From \eqref{Proof:Proposition2Proof12}, the sum of the remaining bits for CTs in $\mathcal{S}$ at slot $n+1$, $\sum_{\ell \in \mathcal{S}}\rho_{n+1}(\ell)$, 
can be expressed in terms of $\sum_{\ell \in \mathcal{S}}\rho_{n}(\ell)$,
{\small
\begin{align}\label{Proof:Corollary3_1}
\sum_{\ell \in \mathcal{S}}\rho_{n+1}(\ell)=\sum_{\ell \in \mathcal{S}}\rho_{n}(\ell)-b_{n}^*
=\frac{\sum_{\ell\in \mathcal{S}}\rho_{n}(\ell)\l(\frac{1}{\zeta_{N-n}(\mathcal{S})}\r)^{\frac{1}{m-1}}    }{(g_{n})^{\frac{1}{m-1}}
+\l(\frac{1}{\zeta_{N-n}(\mathcal{S})}\r)^{\frac{1}{m-1}}    }.
\end{align}}
Substituting  \eqref{Proof:Corollary3_1} into \eqref{Proposition:OptmalPrefetcher2} gives  
$\eta_{n+1}=\eta_{n}\Pi_{n}$,
where
{\small
\begin{align}\label{Proof:Corollary3_3}
&\Pi_{n}=\nonumber\\
&\l\{
\begin{aligned}
&\frac{\l(\frac{1}{\zeta_{N-n-1}(\mathcal{S})}\r)^{\frac{1}{m-1}}    }
{(g_{n+1})^{\frac{1}{m-1}}+\l(\frac{1}{\zeta_{N-n-1}(\mathcal{S})}\r)^{\frac{1}{m-1}}   },
&& \!\!\!\!\textrm{$n<N_P-1$,}\\
&\frac{\l(\frac{1}{\xi_{N-N_P}}\r)^{\frac{1}{m-1}} \sum_{\ell \in \mathcal{S}}p(\ell)^{-\frac{1}{m-1}}}{(g_{N_P})^{\frac{1}{m-1}}\!\!+\!\!\l(\frac{1}{\xi_{N-N_P}}\r)^{\frac{1}{m-1}}\sum_{\ell\in \mathcal{S}}p(\ell)^{-\frac{1}{m-1}}},
&& \!\!\!\!\textrm{$n=N_P-1.$}\\
\end{aligned}
\r.
\end{align}}
Noting that $g_{n}>0$, the coefficients $\{\Pi_{n}\}$ of \eqref{Proof:Corollary3_3} are always less than one and
the condition $\eta_{n+1}>\eta_n$ is satisfied, completing the proof of the first property. 
Due to the first property, 
the first threshold $\eta_1$ is the largest. In addition, it is obvious that $\eta_1$ is maximized when $g_1\rightarrow \infty$,
\begin{align}
\lim_{g_1\rightarrow 0 }\eta_1=&\frac{\sum_{k\in \mathcal{S}}\gamma(\ell)}{\sum_{\ell \in \mathcal{S}}p(\ell)^{-\frac{1}{m-1}}}\nonumber\\
&\overset{(a)}{\leq} \frac{1}{|\mathcal{S}|}\sum_{\ell \in\mathcal{S}}\gamma(\ell) p(\ell)^{\frac{1}{m-1}}\leq \max_{\ell=1,\cdots, L}\l[\gamma(\ell) p(\ell)^{\frac{1}{m-1}}\r],\nonumber
\end{align}
where (a) follows from Jensen's inequality. We complete the proof of the second property.

\subsection{Proof of Lemma \ref{Lemma:Threshold}}\label{App:Threshold}
After transmitting $b_n^*$ of \eqref{Proof:Proposition1_3}, the sum of elements of the status vector $\sum_{\ell \in \mathcal{S}} \rho_{n+1}(\ell)$ becomes 
{\small
\begin{align}
\sum_{\ell \in \mathcal{S}} \rho_{n+1}(\ell)=\sum_{\ell \in \mathcal{S}} \rho_{n}(\ell)-b_n^*
=\frac{\sum_{\ell \in \mathcal{S}}\rho_n(\ell)\l(\frac{1}{\zeta_{N-n}(\mathcal{S})}\r)^{\frac{1}{m-1}}}
{(g_n)^{\frac{1}{m-1}}+\l(\frac{1}{\zeta_{N-n}(\mathcal{S})}\r)^{\frac{1}{m-1}}}\nonumber
\end{align}}and $\sum_{\ell \in \mathcal{S}}\rho_{N_P}(\ell)$ is expressed as a cascade form, 
{\small
\begin{align}
\!\!\sum_{\ell \in \mathcal{S}}\rho_{N_P}(\ell)\!\!&=
\l(\sum_{\ell \in \mathcal{S}}\rho_n(\ell)\r) \cdot
\prod_{k=n}^{N_P-1}
\frac{ \l(\frac{1}{\zeta_{N-k}(\mathcal{S})}\r)^{\frac{1}{m-1}}}
{(g_k)^{\frac{1}{m-1}}+\l(\frac{1}{\zeta_{N-k}(\mathcal{S})}\r)^{\frac{1}{m-1}}}.\nonumber
\end{align}}Inserting the above into $\eta_{N_P}$ in \eqref{Proposition:OptmalPrefetcher2}, we have the result.

\subsection{Proof of Proposition~\ref{Theorem:PrefetchingGain:FF}}\label{App:PrefetchingGain:FF}

As shown in Appendix~\ref{App:PrefetchingGain:SF}, the prefetching gain $G_P$ is minimized when $\gamma(\ell)=\frac{\Gamma}{L}$ and $p(\ell)=\frac{1}{L}$ for all $\ell$ given~as
\begin{align}
{G_P}\geq \frac{\lambda \xi_{N-N_P} \l(\frac{\Gamma}{L}\r)^m}{\lambda \Gamma^m \zeta_N(\mathcal{S}) }=\frac{\xi_{N-N_P}}{\zeta_N(\mathcal{S})L^m},\nonumber
\end{align}
where the prefetching-task set includes all CTs. Using \eqref{Eq:DemandFetchingEnergyBound:FF} enables the following expression of~$\xi_{N-N_P}$,
\begin{align}
\xi_{N-N_P}=\epsilon_1 (N-N_P)^{-(m-1)},\nonumber
\end{align}
where the coefficient $\epsilon_1$ exists in the range of $\frac{1}{\mathbb{E}\l[{g}\r]} \leq \epsilon_1\leq \mathbb{E}\l[\frac{1}{g}\r]$ and its first and second equalities hold when $N-N_P\rightarrow \infty$ or $N-N_P=1$, respectively.     
Similarly,  the coefficient $\zeta_N(\mathcal{S})$ 
is expressed as follows: 
\begin{align}
\zeta_N(\mathcal{S})=\epsilon_2{ \l[N_P +(N-N_P)L^{\frac{m}{m-1}}\r]^{-(m-1)}},\nonumber
\end{align}
where the coefficient $\epsilon_2$ exists in the same range as $\epsilon_1$ and the first and second equality hold when $N\rightarrow \infty$ or $N=1$, respectively.
It is worth mentioning that each of the coefficient becomes closer to 
$\frac{1}{\mathbb{E}\l[{g}\r]}$ as the corresponding fetching duration increases. It is straightforward to show that $\epsilon_1>\epsilon_2$ due to the fact that the fetching period of the conventional fetching is $N-N_P$ slots whereas that of live prefetching is $N$ slots. As a result, the prefetching gain $G_P$ is
\begin{align}
G_P&\geq \frac{\epsilon_1}{\epsilon_2}\l[\frac{N-N_P (1-L^{-\frac{m}{m-1}})}{N-N_P }\r]^{m-1}\nonumber\\
&>\l[\frac{N-N_P (1-L^{-\frac{m}{m-1}})}{N-N_P }\r]^{m-1}. \nonumber
\end{align}

\bibliographystyle{ieeetran}
\bibliography{IEEEabrv,Wireless_Prefetching_related_works}

\begin{IEEEbiography}
[{\includegraphics[width=1in,clip,keepaspectratio]{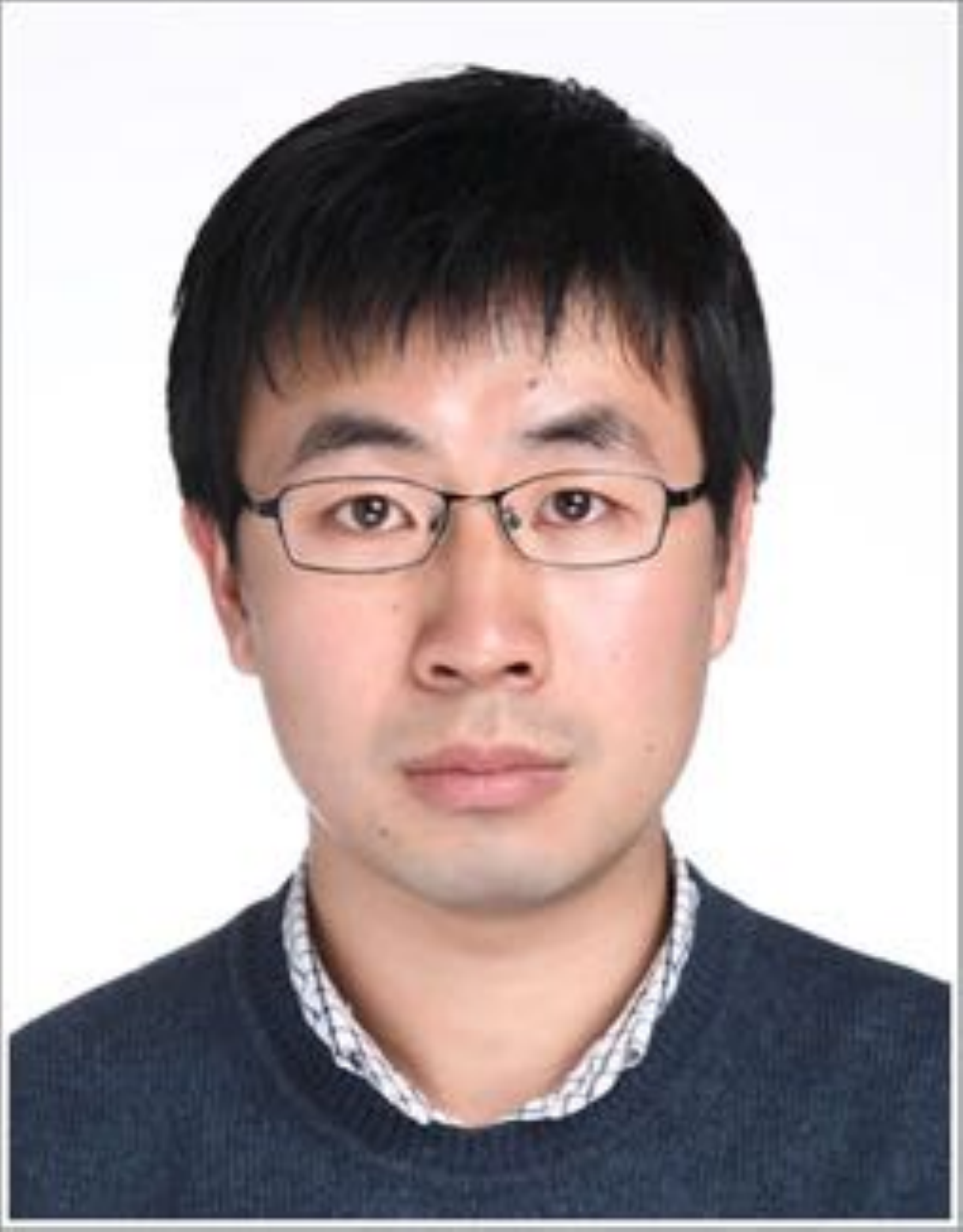}}]
{Seung-Woo Ko}
received the B.S, M.S, and Ph.D degrees at the School of Electrical $\&$ Electronic Engineering, Yonsei University, Korea, in 2006, 2007, and 2013, respectively.
Since Apr.~2016, he has been a postdoctoral researcher in the Dept. of Electrical and Electronics Engineering (EEE) at The University of Hong Kong.
He had been a senior researcher in LG Electronics, Korea from Mar.~2013 to Jun.~2014 and a postdoctoral researcher at Yonsei University, Korea from Jul.~2014 to Mar.~2016.  
His research interests focus on  radio resource management in 5G network architecture, with special emphasis on mobile computation offloading and edge cloud computing. 
\end{IEEEbiography}

\begin{IEEEbiography}
[{\includegraphics[width=1in,clip,keepaspectratio]{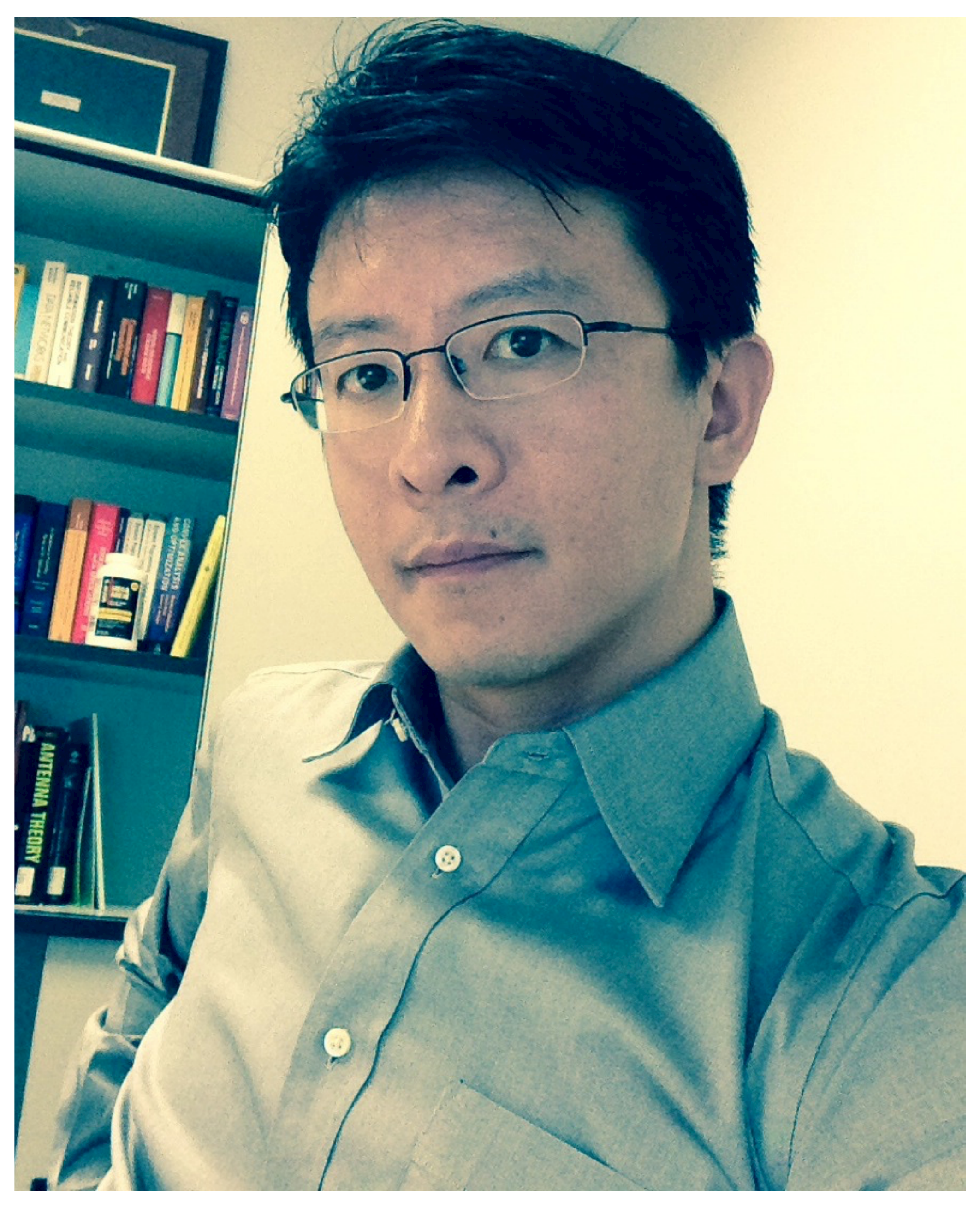}}]
{Kaibin Huang} received the B.Eng. (first-class hons.) and the M.Eng. from the National University of Singapore, respectively, and the Ph.D. degree from The University of Texas at Austin (UT Austin), all in electrical engineering.

Since Jan. 2014, he has been an assistant professor in the Dept. of Electrical and Electronic Engineering (EEE) at The University of Hong Kong. He is an adjunct professor in the School of EEE at Yonsei University in S. Korea. He used to be a faculty member in the Dept. of Applied Mathematics (AMA) at the Hong Kong Polytechnic University (PolyU) and the Dept. of EEE at Yonsei University. He had been a Postdoctoral Research Fellow in the Department of Electrical and Computer Engineering at the Hong Kong University of Science and Technology from Jun. 2008 to Feb. 2009 and an Associate Scientist at the Institute for Infocomm Research in Singapore from Nov. 1999 to Jul. 2004. His research interests focus on the analysis and design of wireless networks using stochastic geometry and multi-antenna techniques.

He frequently serves on the technical program committees of major IEEE conferences in wireless communications. He has been the technical chair/co-chair for the IEEE CTW 2013, the Comm. Theory Symp. of IEEE GLOBECOM 2014, and the Adv. Topics in Wireless Comm. Symp. of IEEE/CIC ICCC 2014 and has been the track chair/co-chair for IEEE PIMRC 2015, IEE VTC Spring 2013, Asilomar 2011 and IEEE WCNC 2011. Currently, he is an editor for IEEE Journal on Selected Areas in Communications (JSAC) series on Green Communications and Networking, IEEE Transactions on Wireless Communications, IEEE Wireless Communications Letters. He was also a guest editor for the JSAC special issues on communications powered by energy harvesting and an editor for IEEE/KICS Journal of Communication and Networks (2009-2015). He is an elected member of the SPCOM Technical Committee of the IEEE Signal Processing Society. Dr. Huang received the 2015 IEEE ComSoc Asia Pacific Outstanding Paper Award, Outstanding Teaching Award from Yonsei, Motorola Partnerships in Research Grant, the University Continuing Fellowship from UT Austin, and a Best Paper Award from IEEE GLOBECOM 2006 and PolyU AMA in 2013. 
\end{IEEEbiography}

\begin{IEEEbiography}
[{\includegraphics[width=1in,clip,keepaspectratio]{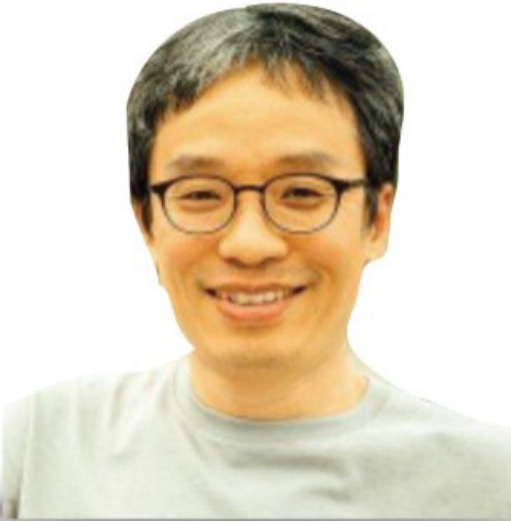}}]
{Seong-Lyun Kim} is a Professor of wireless networks at the School of Electrical \& Electronic Engineering, Yonsei University, Seoul, Korea, heading the Radio Resource Management \& Optimization Laboratory (RAMO) and the Center for Flexible Radio (CFR+). He was an Assistant Professor of Radio Communication Systems at the Department of Signals, Sensors \& Systems, Royal Institute of Technology (KTH), Stockholm, Sweden. He was a Visiting Professor at the Control Group, Helsinki University of Technology  (now Aalto), Finland, and the KTH Center for Wireless Systems. He served as a technical committee member or a chair for various conferences, and an editorial board member of IEEE Transactions on Vehicular Technology, IEEE Communications Letters, Elsevier Control Engineering Practice, Elsevier ICT Express, and Journal of Communications and Network. He served as the leading guest editor of IEEE Wireless Communications, and IEEE Network for wireless communications in networked robotics. He also consulted various companies in the area of wireless systems both in Korea and abroad. His research interest includes radio resource management and information theory in wireless networks, economics of wireless systems, and robotic networks. His degrees include BS in economics (Seoul National University), and MS \& PhD in operations research (Korea Advanced Institute of Science \& Technology).
\end{IEEEbiography}

\begin{IEEEbiography}
[{\includegraphics[width=1in,clip,keepaspectratio]{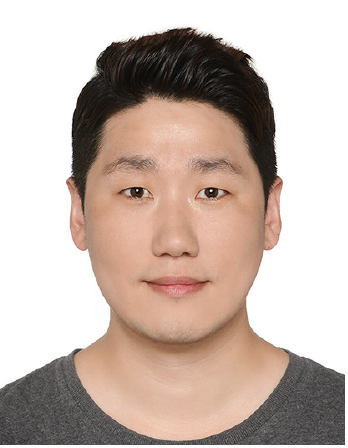}}]
{Hyukjin Chae} received the B.S and Ph.D degree in electrical and electronic
engineering from Yonsei University, Seoul, Korea. He joined LG Electronics,
Korea, as a Senior Research Engineer in 2012. His research interests
include interference channels, multiuser MIMO, D2D, V2X, and full duplex
radio. From Sep. 2012, he has contributed and participated as a delegate in
3GPP RAN1 with interests in FD MIMO, positioning enhancement, D2D, and V2X communications.
\end{IEEEbiography}

\end{document}